\documentclass[11pt,onecolumn]{IEEEtran}
\usepackage{cite}
\ifCLASSINFOpdf
  \usepackage[pdftex]{graphicx}
\else
  \usepackage[dvips]{graphicx}
\fi
\usepackage{subcaption}
\usepackage{epstopdf}
\usepackage{amsmath}
\usepackage{amsthm}
\usepackage{algcompatible}
\usepackage{amssymb}
\usepackage{enumerate}
\usepackage[usenames, dvipsnames]{color}
\usepackage{algorithm}
\usepackage{algpseudocode}
\usepackage{mathtools}
\usepackage{breqn}
\usepackage{bbm}
\usepackage{bm}
\usepackage{cite}
\usepackage{pdfpages}
\usepackage{cuted}
\usepackage{multicol}

\newtheorem{definition}{\bf Definition}
\newtheorem{assumption}{\bf Assumption}
\newtheorem{remark}{\bf Remark}
\newtheorem{theorem}{\bf Theorem}
\newtheorem{lemma}{\bf Lemma}
\newtheorem{proposition}{\bf Proposition}

\usepackage{array}
\newcolumntype{"}{@{\hskip\tabcolsep\vrule width 2pt\hskip\tabcolsep}}

\newcommand{\mr}{\mathrm}
\newcommand{\mc}{\mathcal}
\newcommand{\mbf}{\mathbf}
\newcommand{\mbb}{\mathbb}

\newcommand{\mbZ}{\mathbb Z}
\newcommand{\mbR}{\mathbb R}
\newcommand{\mbC}{\mathbb C}

\newcommand{\mcR}{\mathcal R}

\newcommand\numberthis{\addtocounter{equation}{1}\tag{\theequation}}

\DeclarePairedDelimiter{\diagfences}{(}{)}

\newcommand{\blockdiag}{\operatorname{blockdiag}\diagfences}

\title{\LARGE \bf
Data-driven control on encrypted data
}

\author{Andreea B. Alexandru, 
Anastasios Tsiamis and 
George J. Pappas%
\thanks{The authors are with the Department of Electrical and Systems Engineering, University of Pennsylvania, Philadelphia, PA 19104.
        {\tt\small \{aandreea,atsiamis,pappasg\}@seas.upenn.edu}}%
}

\begin{document}

\maketitle
\thispagestyle{plain}
\pagestyle{plain}

\begin{abstract}
We provide an efficient and private solution to the problem of encryption-aware data-driven control. We investigate a Control as a Service scenario, where a client employs a specialized outsourced control solution from a service provider. The privacy-sensitive model parameters of the client's system are either not available or variable. Hence, we require the service provider to perform data-driven control in a privacy-preserving manner on the input-output data samples from the client. To this end, we co-design the control scheme with respect to both \emph{control performance} and \emph{privacy specifications}. 
First, we formulate our control algorithm based on recent results from the behavioral framework, and we prove closeness between the classical formulation and our formulation that accounts for noise and~precision errors arising from encryption. Second, we use a state-of-the-art leveled homomorphic encryption scheme to enable the service provider to perform high complexity computations on the client's encrypted data, ensuring privacy. 
Finally, we streamline our solution by exploiting the rich structure of data, and meticulously employing ciphertext batching and rearranging operations to enable parallelization. This solution achieves more than twofold runtime and memory improvements compared to our prior work. 
\end{abstract}

\IEEEpeerreviewmaketitle

\section{Introduction}
\label{sec:introduction}

As cloud services become pervasive, inspiring the term ``Everything as a Service"~\cite{Xaas}, privacy of the underlying data and algorithms becomes indispensable. Distributed systems coordination and process control enterprises can benefit from such services by commissioning cloud servers to aggregate and manipulate the increasing amounts of data or implement reference governors. 
For example, in the context of smart building automation, Control as a Service (CaaS) businesses~\cite{smart-buildings}~\mbox{offer} high-performance algorithms to optimize a desired cost, while achieving the required control goals. Clients like hospitals, factories, commercial and residential buildings might prefer such specialized outsourced solutions rather than investing~in local control solutions. 
However, the data that is outsourced, stored and processed at a CaaS provider's cloud server consists of privacy-sensitive measurements such as user patterns, energy and temperature measurements, occupancy information etc. Nowadays, data can be monetized and used to disrupt the normal functionality of a system, leading to frequent database leakages and cyberattacks. This compels researchers and practitioners to design decision algorithms that are optimized for both efficiency and privacy. In this work, we provide a proof of concept for such a private cloud-based control algorithm.

In a client-server context, homomorphic encryption provides an appealing tool that allows one server to evaluate a functionality over the encrypted data of the client without the need to decrypt the data locally, at independently specified precision and security levels. In contrast, secure multi-party computation tools require collaboration between multiple cloud servers to execute the encrypted computations, and differential privacy introduces a trade-off between the accuracy and privacy levels. As discussed in~\cite{vanDijk2010impossibility}, private single client computing can be enabled using homomorphic encryption schemes. Moreover, this technology has matured enough for it to be deployed~in~practice, especially in the healthcare~\cite{iDash} and financial~\cite{Masters2019towards} sectors. Since its genesis in~\cite{GentryPhD}, fully homomorphic encryption (FHE) has been developed substantially, in terms of more theoretically efficient leveled constructions~\cite{Brakerski11,Gentry2013homomorphic,Cheon2017homomorphic}, bootstrapping methods~\cite{Cheon2018bootstrapping, Chen2019improved}, computational and hardware optimizations. There are many libraries that implement various FHE schemes and capabilities~\cite{Halevi2014algorithms,sealcrypto,PALISADE,TFHE} and a tremendous number of papers that build on them.

\subsection{Related work}\label{subsec:related}
Recently, private control algorithms operating on encrypted data have received a lot of interest~\cite{Kim16encrypting,Farokhi2017secure,Darup18towards,Alexandru2018cloud,Alexandru2019encrypted,Suh2020sarsa}. These works~can be classified with respect to the type of control algorithms they implement and the cryptographic tools they use. Linear control algorithms with public model and gains are considered in most works and are implemented using partially homomorphic encryption (PHE) schemes. 
Nonlinear control algorithms with public model are considered in~\cite{Darup18towards,Alexandru2018cloud}, where PHE is combined either with secure multiparty computation schemes or the client partakes in the computation. In~\cite{Kim16encrypting,Alexandru2019encrypted}, somewhat or fully homomorphic encryption schemes are preferred, either to guarantee more privacy for the system parameters (that are still known at the client) or to achieve more complex control algorithms. 
Like our current work,~\cite{Suh2020sarsa} is based on data-driven techniques and uses leveled homomorphic encryption. There, the cloud computes the value function in a reinforcement learning task over multiple time steps in a one-shot way, while we collect encrypted data from the client at every time step and compute on it iteratively. 
Finally, our current work enhances and improves~\cite{Alexandru2020towards} in the following aspects: we propose a new encrypted solution that reduces more than twofold the time and the memory requirements of the previous solution, allowing us to simulate more realistic and more secure scenarios; we prove the close relationship between the standard data-driven LQR and our reformulation; and we present many more extensions.

As mentioned before, most private cloud-based control schemes assume that the model of the system is known to the client or public. However, the model might not be always available and has to be computed from data. 
Identification and data-driven control of unknown systems have been extensively studied in classical and recent literature \cite{Ljung1999system,van2012subspace,Favoreel1999model,Willems2005note,markovsky2008data,de2019persistency,Salvador2019data,Coulson2019data,Baggio2019data,Berberich2020data}. These methods are usually designed with the objective of sample-efficiency and control performance, without considering a private implementation. 
For example, the standard system identification-certainty equivalence control architecture might require solving non-convex problems~\cite{yu2018identification} or might involve Singular Value Decomposition (SVD)~\cite{van2012subspace}, which are prohibitive from an encrypted evaluation perspective.
On the other hand, the behavioral framework~\cite{Willems2005note,markovsky2008data,de2019persistency,Coulson2019data,Berberich2020data,VanWaarde2020data}, which has received renewed interest recently, is more compatible with modern encryption tools and allows us to use less costly encrypted operations. 
In the behavioral framework, an alternative data-driven representation of linear systems is considered: if the collected input-output data are rich enough (input is persistently exciting), then the model is implicitly characterized by the collected data~\cite{Willems2005note}.

\subsection{Contributions}
\label{subsec:contributions}

Our goal is to perform online data-driven control on encrypted data, while maintaining the privacy of the client's uploaded input-output data, desired setpoint and control actions. In the context of private implementation, we need to depart from the typical techniques used in the non-private versions and make compromises between tracking performance and privacy. 
Our controller is based on the behavioral framework and the control performance is captured by the LQR cost. We consider two versions and adapt them to account for an encrypted implementation: i) an offline version, where precollected data from the system is encrypted and handed to the cloud to compute an offline feedback control law, and ii) an online version, where the cloud server computes the control at every time step based on both the offline precollected data and on the encrypted measurements received from the client. 
We employ a new leveled homomorphic encryption scheme optimized for efficiency and precision as the powerful base tool for encrypting and evaluating the private data on the untrusted cloud machine. At the same time, we rewrite the computation of the control input as a low-depth arithmetic circuit by using mathematical tools such as the Schur complement to facilitate matrix inversions, and carefully batching the private values into ciphertexts. 
Specifically, our contributions are the following:
\begin{itemize}
	\item Formulate a setpoint tracking problem that learns the control action from the available data at every time step and is amenable to encrypted implementation.
	\item Prove the closeness between this formulation and a data-driven LQR problem.	
	\item Extend this approximation to allow the online collection and incorporation of samples to improve robustness and adaptation to new data and to errors arising from the use of encryption. 
	\item Design storage and computation efficient encrypted versions of the offline and online data-driven problems, by carefully encoding the values to facilitate parallelization and re-engineering the way the operations are performed. 
	\item Implement the algorithms for a standard security parameter and realistic client and server machines. Present the results for a temperature control problem and showcase the storage and complexity improvements compared to the previous algorithm in~\cite{Alexandru2020towards}. 
\end{itemize}

We emphasize that the techniques for performing the encrypted computations efficiently are one of main contributions of the paper. These tools can easily be generalized and used in implementing other related problems in an encryption-friendly way, such as learning and adaptive control algorithms, recursive least squares, sequential update of inverses etc.%

\subsection{Notation}
For a vector $\mbf x\in\mbb R^n$, we denote by $\mbf x_{[i]}$ the $i$'th element of the vector. We denote the $L^2$ norm of a vector $\mbf x$ by $\lVert \mbf x\rVert_2$ and the spectral norm of a matrix $\mbf{M}$ by $\lVert\mbf{M}\rVert_{2}$. For a matrix $\mbf P\in \mbb S^n_{++}$ the $\mbf P$-quadratic norm is $\lVert\mbf x\rVert^2_{\mbf P} := \mbf x^\intercal \mbf P \mbf x$. The vector $\mbf e_i$ is the vector of all zeros, except for position $i-1$ where it has 1. For a set of consecutive integers, we define $[n]:=\{0,\ldots,n-1\}$. The notation $\mr{E}(\cdot)$ denotes the encryption of a quantity, with more parameters to be specified in the text. 

\section{Problem formulation}
\label{sec:formulation}
Let $\mbf x\in\mbb R^n, \mbf u \in \mbb R^m, \mbf y \in \mbb R^p$ be the state, control input and measurement of a linear system in~\eqref{eq:system}. 
A client owns system~\eqref{eq:system} and contracts a cloud service to provide the control for this system: 
\begin{align}\label{eq:system}
\begin{split}
	\mbf x_{t+1} &= \mbf A \mbf x_t + \mbf B \mbf u_t, \quad \mbf y_t = \mbf C \mbf x_t, 
\end{split}
\end{align}
as depicted in Figure~\ref{fig:diagram}.
A private CaaS should solve the control problem that minimizes the cost associated to setpoint tracking while satisfying the unknown system's dynamics~\eqref{eq:system}. 

\begin{figure}[ht]
	\centering
	\includegraphics[width=0.35\columnwidth]{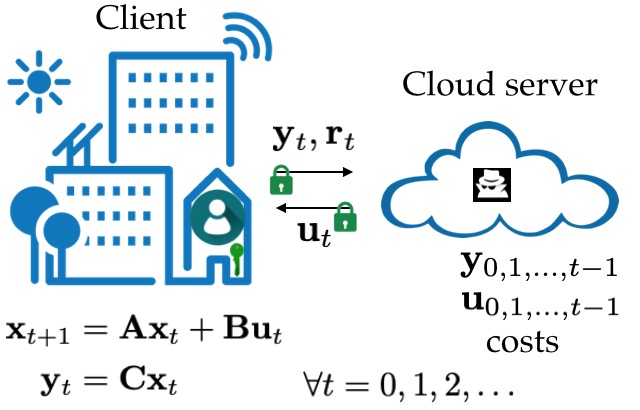}
	\caption{Client-server diagram for data-driven control.}
	\label{fig:diagram}
\end{figure}

\noindent\textbf{Control requirements. }
Assume we are given a batch of offline input-output data for system~\eqref{eq:system}. 
Then, at every time $t$, given $\mbf u_{0:t-1}$ and $\mbf y_{0:t-1}$ we want to compute in a private and receding horizon fashion $\mbf u^{\ast,t}\in\mbb R^{Nm}$ in order to track the reference $\mbf r_t$, for some costs $\bar{\mbf Q}, \bar{\mbf R}$, which is the solution of the LQR optimization problem:
\begin{align}\label{eq:basicLQR}
\begin{split}
\min_{\mbf u,\mbf y}~&~\frac{1}{2} \sum_{k=t}^{N+t-1} \left( || \mbf y_k - \mbf r_{k}||_{\bar{\mbf Q}}^2 + || \mbf u_k ||_{\bar{\mbf R}}^2 \right) \\
s.t.~&~\mbf x_{k+1}=\mbf A \mbf x_{k}+ \mbf B \mbf u_k,\, \mbf y_k=\mbf C \mbf x_k,
\end{split}
\end{align}
with no prior knowledge about the system model $\mbf A, \mbf B, \mbf C$. 
We also want~\eqref{eq:basicLQR} to perform well under unknown small process and measurement noise and bias. 

In this work, we focus on investigating an approximation of problem~\eqref{eq:basicLQR} without hard constraints on the inputs or outputs of the system. In Section~\ref{sec:reformulation}, we describe a data-driven reformulation of~\eqref{eq:basicLQR} that is encryption-friendly. 

\noindent\textbf{Privacy requirements. }
In the scenario we consider, the cloud service should not be able to infer anything about the client's private data, which consists of the input signals $\mbf u$, the output signals $\mbf y$, the model $\mbf A,\mbf B,\mbf C$ and state $\mbf x$ (the last four quantities being unknown in a data-driven control problem), and any intermediate values. The costs $\bar{\mbf Q}, \bar{\mbf R}$ can be chosen by the cloud, as part of the CaaS service or chosen by the client. 
The cloud service is considered to be \textit{semi-honest}, which means that it does not deviate from the client's specifications. We assume this to be the case because the cloud server is under contract. 

The formal privacy definition is provided in Section~\ref{sec:privacy} and can be written in the standard simulation paradigm \cite[Ch.~7]{Goldreich04foundationsII}, i.e., all information obtained by the server after the execution of the protocol (while also keeping a record of the intermediate computations) can be obtained solely from the inputs and outputs available to the server.

\noindent\textbf{Efficiency requirements. }
The computation time of the control action at the current time step should not exceed the sampling time, which is the time until the next measurement is fetched. In the private case, we expect a large overhead for complex controllers and consider applications with sampling time of the order of minutes. 
Moreover, we require the client to be exempt from heavy computation and communication, and the bulk of the computation (in the allowable time) should be performed at the server side. 

\section{Encryption-aware learning formulation}
\label{sec:reformulation}

We can formulate this control problem inspired by the behavioral and subspace identification frameworks explored in\cite{Willems2005note,markovsky2008data,de2019persistency,Coulson2019data,Berberich2020data}.
First, we introduce some preliminary concepts. 
A \textit{block-Hankel matrix} for the input signal $\mbf u = \left[ \begin{matrix}  \mbf u_0^\intercal & \mbf u_2^\intercal &\ldots& \mbf u_{T-1}^\intercal  \end{matrix} \right]^\intercal \in\mathbb R^{mT}$ is given by the following, for a positive integer $L$:
\begin{equation}\label{eq:HankelU}
	H_L\mbf U := \left[\begin{matrix}	\mbf u_0		& \mbf u_1 	& \ldots & 	\mbf u_{T-L}	\\
						   			\mbf u_1		& \mbf u_2 	& \ldots & \mbf	u_{T-L+1}	\\ 
						  			\vdots		& 			&\ddots &	\vdots	\\
						   			\mbf u_{L-1}	& \mbf u_{L}	& \ldots & \mbf u_{T-1}\end{matrix} \right].
\end{equation} 
Furthermore, the signal $\mbf u$ is persistently exciting of order $L$ if $H_L\mbf U\in\mathbb R^{mL \times (T-L+1)}$ is full row rank.

Let us construct block-Hankel matrices for the ``past" and ``future" input and output data, $\mbf u \in \mathbb R^{mT}$, respectively, $\mbf y \in \mathbb R^{pT}$, for $M$ samples for the past data and $N$ samples for the future data, and $S:=T-M-N+1$:
\begin{align}\label{eq:past+future}
\begin{split}
\mbf U^p :=& \left[\begin{matrix}\mbf I_{mM} & \mbf 0_{mM\times mN}\end{matrix}\right] \cdot H_{M+N} \mbf U\in\mathbb R^{mM \times S },\qquad \mbf U^f := \left[\begin{matrix}\mbf 0_{mN\times mM} & \mbf I_{mN}\end{matrix}\right] \cdot  H_{M+N}\mbf U\in\mathbb R^{mN \times S },\\
\mbf Y^p :=& \left[\begin{matrix}\mbf I_{pM} & \mbf 0_{pM\times pN}\end{matrix}\right] \cdot H_{M+N}\mbf Y\in\mathbb R^{pM \times S },\hspace{1.2cm}
\mbf Y^f := \left[\begin{matrix}\mbf 0_{pN\times pM} & \mbf I_{pN}\end{matrix}\right] \cdot H_{M+N}\mbf Y \in\mathbb R^{pN \times S }.
\end{split}
\end{align}

\subsection{Offline feedback data-driven control problem}
\label{subsec:one_shot}

Assume that we are given precollected input and output data, with $\mbf U^p, \mbf Y^p, \mbf U^f, \mbf Y^f$ the respective past and future Hankel matrices, for some past and future horizons  $M,N$. Assume also that the precollected input is persistently exciting. 
\begin{assumption}[Data richness]
\label{assum:richness}
We assume the offline precollected input trajectory $\mbf u$ is persistently exciting of order $M+N+n$, where $n$ is the order of the system~\cite{Willems2005note}. 
\end{assumption}

Fix a time $t$ and let $\bar{\mbf u}_t=\mbf{u}_{t-M:t-1}$ be the batch vector of the last $M$ inputs. The batch vector of the last $M$ outputs $\bar{\mbf y}_t$ is defined similarly. If $M\ge n$, then the standard LQR problem~\eqref{eq:basicLQR} can be re-formulated as
the following data-driven control problem~\cite{Coulson2019data}:
\begin{align}\label{eq:optimization0}
\begin{split}
\min_{\mbf g,\mbf u,\mbf y}~&~\frac{1}{2} \sum_{k=t}^{N+t-1} \left( || \mbf y_k - \mbf r_{k}||_{\bar{\mbf Q}}^2 + || \mbf u_k ||_{\bar{\mbf R}}^2 \right) \\
s.t.~&~\left[\begin{matrix} \mbf U^p \\ \mbf Y^p \\ \mbf U^f \\ \mbf Y^f \end{matrix}\right] \cdot \mbf g = \left[\begin{matrix} \bar{\mbf u}_t \\ \bar{\mbf y}_t \\ \mbf u \\ \mbf y \end{matrix}\right],
\end{split}
\end{align}
where the state representation has been replaced with the precollected data and $g$ represents the preimage of the system's trajectory with respect to the precollected block-Hankel matrices.
Note that $\mbf u^{\ast,t}_{[1:m]}$, the first $m$ elements of $\mbf u^{\ast,t}$, will be input into the system in a receding horizon fashion, and $\mbf y^{\ast,t}$ is the predicted output.

We now depart from the behavioral control problem~\eqref{eq:optimization0} considered in the existing literature and explore a more encryption-friendly form.

First, we rewrite~\eqref{eq:optimization0} as a minimization problem depending only on $\mbf g$ by enforcing $\mbf u = \mbf U^f \mbf g$ and $\mbf y = \mbf Y^f \mbf g$.  
Second, in practice, there will be noise affecting the output measurement, as well as precision errors induced by encryption, which might prevent an exact solution to the equality constraint of~\eqref{eq:optimization0}. 
Hence, we prefer a least-squares penalty approach to the equality constraint in~\eqref{eq:optimization0} with regularization weights $\lambda_y$ and $\lambda_u$.
Finally, to reduce overfitting and precision errors due to encryption, we also penalize the magnitude of $\mbf g$ through two-norm regularization.~We opt for two-norm regularizations to obtain better efficiency for the encrypted algorithm, as well as more robustness with respect to noise, and uniqueness of the solution $\mbf g^{\ast,t}$.  This yields the following formulation, where $\mbf Q=\blockdiag{\bar{\mbf Q},\dots,\bar{\mbf Q}}$ and $\mbf R=\blockdiag{\bar{\mbf R},\dots,\bar{\mbf R}}$ and we reuse the notation $\mbf r_t$ for the batch reference signal: 
\begin{align}\label{eq:optimization2}
\begin{split}
\min_{\mbf g}~& \frac{1}{2} \left( ||\mbf Y^f\mbf g-\mbf r_t||_{\mbf Q}^2 + ||\mbf U^f\mbf g||_{\mbf R}^2 + \lambda_y || \mbf Y^p \mbf g - \bar{\mbf y}_t ||_2^2 + \lambda_u || \mbf U^p \mbf g - \bar{\mbf u}_t ||_2^2 + \lambda_g || \mbf g ||_2^2 \right).
\end{split}
\end{align}
Note that the resulting problem~\eqref{eq:optimization2} is an approximation of~\eqref{eq:optimization0}. Finally,~\eqref{eq:optimization2} can be written as a quadratic program:
\begin{align}\label{eq:optimization3}
\begin{split}
\min_{\mbf g}~&~ \frac{1}{2} \mbf g^\intercal \mbf M \mbf g - \mbf g^\intercal \left({\mbf Y^f}^\intercal \mbf Q \mbf r_t + \lambda_y {\mbf Y^p}^\intercal \bar{\mbf y}_t +\lambda_u {\mbf U^p}^\intercal \bar{\mbf u}_t   \right),
\end{split}
\end{align}
where $\mbf M\in\mbR^{S\times S}$ is:
\begin{equation}\label{eq:M}
\mbf M:= {\mbf Y^f}^\intercal \mbf Q\mbf Y^f + {\mbf U^f}^\intercal \mbf R\mbf U^f + \lambda_y{\mbf Y^p}^\intercal \mbf Y^p + \lambda_u{\mbf U^p}^\intercal \mbf U^p + \lambda_g \mbf I.
\end{equation}

Since~\eqref{eq:optimization3} is a strongly convex optimization problem (ensured by the regularization term $\lambda_g \mbf I$), we can find the optimal value for $\mbf g$ by zeroing the gradient of the objective function:
\begin{equation}\label{eq:opt_g}
\mbf g^{\ast,t} = \mbf M^{-1} \left ( {\mbf Y^f}^\intercal \mbf Q \mbf r_t + \lambda_y{\mbf Y^p}^\intercal \bar{\mbf y}_t + \lambda_u{\mbf U^p}^\intercal \bar{\mbf u}_t  \right ).
\end{equation}

Going back to the control input for the current time step we want to compute, we obtain from~\eqref{eq:opt_g}:
\begin{equation}\label{eq:opt_u}
\mbf u^{\ast,t} = \mbf U^f\mbf M^{-1} \left ( {\mbf Y^f}^\intercal \mbf Q \mbf r_t + \lambda_y{\mbf Y^p}^\intercal \bar{\mbf y}_t + \lambda_u{\mbf U^p}^\intercal \bar{\mbf u}_t  \right ),
\end{equation}
from which we select the first $m$ elements and input them to the unknown system.

As seen from~\eqref{eq:opt_u}, the controller has the form of a dynamic output-feedback law, where the feedback terms are computed using only the offline precollected data. 
More details about the choices we made in assembling this data-driven control problem formulation are given in Section~\ref{sec:FHE} and Appendix~\ref{app:theory}.

\subsection{Closeness between approximate and original problems}
We prove that if we select small enough regularization coefficient $\lambda_g$ and large enough penalty coefficients $\lambda_y,\lambda_u$, then the solution of the approximate problem~\eqref{eq:optimization2} is very close to the minimum norm solution of the behavioral control problem~\eqref{eq:optimization0}.
The main challenge in the proof is that the past and future Hankel matrices are rank-deficient. Hence the behavioral control problem~\eqref{eq:optimization0} involves singular matrices in both the objective and the constraints, and has multiple optimal solutions. This requires an involved analysis, where we deal with pseudo-inverses and subspaces, shown in Appendix~\ref{app:theory}. 
To formally state the result we need two definitions.
The set of optimal solutions of~\eqref{eq:optimization0} is denoted by $ \mathcal{G}_{\mathrm{opt}}$:
\begin{equation}\label{eq:set_optimal_solutions}
    \mathcal{G}_{\mathrm{opt}}:= \left\{\mbf{g}:\:\mbf{g} \text{ solves~\eqref{eq:optimization0}}\right\}.
\end{equation}
The minimum norm element of $ \mathcal{G}_{\mathrm{opt}}$ is defined as
\begin{align}
    \mbf{g}_{\min}&:= \arg\min_{\mbf{g}\in\mathcal{G}_{\mathrm{opt}}} ||\mbf{g}||_2 \label{eq:minimum_norm_g}.
\end{align}

\begin{theorem}\label{thm:closeness}
Consider the original behavioral control problem~\eqref{eq:optimization0} and its approximation~\eqref{eq:optimization2}, with penalty coefficient $\lambda_y=\lambda_u=\lambda>0$ and regularization coefficient $\lambda_g>0$. Let $\mbf g_{\min}$ be the minimum norm solution of the behavioral problem~\eqref{eq:optimization0} as defined in~\eqref{eq:minimum_norm_g}. Let $\mbf g^*$ be the optimal solution of the approximate problem~\eqref{eq:optimization2}. Then:
\begin{equation}
    \left\lVert\mbf{g}_{\min}-\mbf{g}^*\right\rVert_2\rightarrow 0,
\end{equation}
as $(\lambda_g,\lambda)\rightarrow (0,\infty)$ restricted on the set $\lambda_g>0,\lambda>0$.
\end{theorem}
Although the behavioral problem has infinite solutions, due to the regularization term the approximate problem~\eqref{eq:optimization2} will only return a solution close to the minimum norm one $\mbf{g}_{\min}$. 

\subsection{Online feedback data-driven control problem}
\label{subsec:online}
To satisfy Assumption~\ref{assum:richness}, it is necessary that the precollected input signal has length at least $(m+1)(M+N+n)-1$,~cf.~\cite{Coulson2019data}. In practice, Assumption~\ref{assum:richness} might be violated if less precollected data is available. 
Another issue with the precollected data is that it can be affected by perturbations, e.g., measurement noise.
To alleviate these issues, we prefer an online algorithm, where the Hankel matrices $H_{M+N}\mbf U$ and $H_{M+N}\mbf Y$ are updated at each time step for another $\bar T$ steps with the used control input and the corresponding output measured (while keeping the block Hankel form). This approach empirically ensures richness of the data and robustness to perturbations. However, it is important that $\bar T$ is not too large, in order to prevent overfitting.

Note that the precollected and the online data in this~online algorithm belong to different trajectories. For this reason, we must compute the Hankel matrix for each data set separately, then append them in a single matrix~\cite{van2020willems}. This means that the online adaptation of the matrices $\mbf{U}^p,\mbf{U}^f,\mbf{Y}^p,\mbf{Y}^f$ can only start at time $t=M+N-1$, when we obtain enough data $\mbf u_0,\mbf y_0,\dots, \mbf u_{M+N-1}, \mbf y_{M+N-1}$ to fill the first Hankel matrix column for the online data set. We call this phase between $t=M-1$ and $t=M+N-1$ trajectory concatenation.

The data-driven LQR algorithm is given in Algorithm~\ref{alg:online}.

\begin{algorithm}[ht]
	\caption{Online data-driven control algorithm}
	\label{alg:online}
	\begin{algorithmic}[1] {}
		\Require $\bar{\mbf u}_t$, $\bar{\mbf y}_t$, $\mbf U^p$, $\mbf U^f$, $\mbf Y^p$, $\mbf Y^f$, $\mbf Q$, $\mbf R$, $\lambda_y$, $\lambda_u$, $\lambda_g$, $S = T - M - N + 1$, $\bar T$.
		\Ensure $\mbf u_t$ for $t=0,1,\ldots$.
		\For{$t=0,1,\ldots,M-1$}
		    \State Randomly select and input $\mbf u_t$ and measure the output $\mbf y_t$.%
		\EndFor
		\State Construct $\bar{\mbf u}_M= \mbf u_{0:M-1}$ and $ \bar{\mbf y}_M = \mbf y_{0:M-1}$.
		\For{$t=M,M+1,\ldots,M+N+\bar T-1$}
			\State Solve~\eqref{eq:optimization3} for $\mbf g^{\ast,t}$ and obtain~\eqref{eq:opt_g}.
			\State Compute $\mbf u^{\ast,t} = \mbf U^f\mbf g^{\ast,t}$ and obtain~\eqref{eq:opt_u}.
			\State Input to the system $\mbf u_t = \mbf u^{t,\ast}_{[1:m]}$ and measure the output $\mbf y_t$.
			\State Update $\bar{\mbf u}_t$ and $\bar{\mbf y}_t$ to be the last $M$ components of $\left[ \begin{matrix} \bar{\mbf u}_t^\intercal & \mbf u_t^\intercal\end{matrix} \right]^\intercal$ and $\left[ \begin{matrix} \bar{\mbf y}_t^\intercal & \mbf y_t^\intercal \end{matrix} \right]^\intercal$, respectively.
		 \If{$t=M+N-1$}
			    \State Add $\mbf u_{0:t}$ and $\mbf y_{0:t}$ to the $S+1$'th column of the trajectory Hankel matrices $H_{M+N}\mbf U$ and $H_{M+N}\mbf Y$. 
    		\ElsIf{$t>M+N-1$}  
    		    \State Set $S=S+1$. Use $\mbf u_t$, respectively $\mbf y_t$, to add a new column to $H_{M+N}\mbf U$, respectively $H_{M+N}\mbf Y$, while keeping the block Hankel matrix form.   		    
    		\EndIf
		\EndFor
		\While{$t \geq M+N+\bar T$}
		    \State Do lines 6--9.
		\EndWhile
	\end{algorithmic}%
\end{algorithm}%
%
\section{Homomorphic Encryption Preliminaries}
\label{sec:FHE}
The decryption primitive of a homomorphic encryption scheme is a homomorphism from the space of encrypted messages, or \textit{ciphertexts}, to the space of unencrypted messages, or \textit{plaintexts}. Homomorphic encryption allows one server to evaluate multivariate polynomial functionalities over the encrypted data of the client. 
An encryption scheme is called partially homomorphic if it supports the encrypted evaluation of either a linear polynomial or a monomial, somewhat or leveled homomorphic if it supports the encrypted evaluation of a polynomial with a finite degree and fully homomorphic if it supports the encrypted evaluation of arbitrary polynomials. 
Leveled homomorphic schemes can be turned into fully homomorphic schemes by a bootstrapping operation. 
We direct the interested reader to the following surveys on homomorphic encryption~\cite{Halevi2017homomorphic,Brakerski2019fundamentals}. 

The common term for such a multivariate polynomial functionality is \textit{arithmetic circuit} and the logarithm of the degree of the polynomial is the \textit{multiplicative depth} of the circuit. 
Each operation evaluated on ciphertexts introduces some noise, which, if it overflows, can prevent correct decryption. Multiplications introduce the most noise, so we focus on the number of sequential multiplications allowed in an instance of a leveled homomorphic scheme, which we call \textit{multiplicative budget}. 

In this paper, we work with a leveled homomorphic encryption scheme. Specifically, we use the version of the CKKS scheme~\cite{Cheon2017homomorphic}, optimized to run on machine word~size of 64-bit integer arithmetic~\cite{Cheon2018full,Halevi2019improved} instead of multiprecision integer arithmetic. We give more details with respect to the primitives of this scheme in Appendix~\ref{app:CKKS}. 
We chose this scheme because it can perform operations on encrypted real numbers with a smaller error than other leveled homomorphic schemes. Each real number is multiplied by a positive integer scaling factor and truncated, as commonly done, but the real advantage of this scheme is that one can remove the extra scaling factor occurring in the result after a multiplication, through a rescaling procedure, at very little error. This manages the magnitude of the underlying plaintexts, which could otherwise cause overflow in a large depth circuit. 

A ciphertext's size grows with the number of sequential multiplications it supports. Each ciphertext, respectively plaintext, is characterized by a \textit{level}, a \textit{scaling depth} and a number of \textit{moduli}. A new ciphertext (freshly encrypted, rather than obtained as the result of operations on other ciphertexts) is on level $L$, scaling depth 1 and has as many moduli as the multiplicative budget minus one. The number $L$ minus the current level corresponds to the number of rescaling operations previously performed on the ciphertext and scaling depth corresponds to the number of multiplications without rescaling that have been performed. 
After a multiplication followed by a rescaling procedure, the number of levels decreases by one, the scaling depth remains the same and one modulus is dropped. 
To avoid some technicalities and parallel the notion of circuit depth, in the rest of the manuscript, when we refer to the \textit{depth} of a ciphertext, we refer to the multiplicative depth of the arithmetic circuit used in order to obtain that ciphertext. If rescaling is done after every multiplication, then the depth will be the scaling depth minus one.

In the rest of the paper, we will analyze the multiplicative depth of a quantity $x$, obtained after evaluating an arithmetic circuit. 
Out of all circuits that output $x$ on the same inputs, we choose to evaluate the one with the minimum depth $d(x)$. 
In turn, this means that the ciphertext encrypting the resulting quantity $\mr{E}(x)$ will have $d(x)+1$ consumed levels. Note that the fewer the number of levels necessary in the computation, the cheaper the operations on ciphertexts are.

In the CKKS scheme, each plaintext is a polynomial in the ring of integers of a cyclotomic field with dimension $\mr{ringDim}$, which enables us to encode multiple scalar values in a plaintext and a ciphertext. Hence, we \textit{obtain ciphertexts that contain the encryption of a vector}. This is a standard practice in lattice-based homomorphic encryption that allows performing single instruction multiple data ($\mr{SIMD}$) operations, and can bring major computation and storage improvements when evaluating an arithmetic circuit. We can pack up to $\mr{ringDim}/2$ values in one plaintext/ciphertext using a Discrete Fourier Transform~\cite{Cheon2017homomorphic}. 
Packing can be thought of as the ciphertext having $\mr{ringDim}/2$ independent data slots. Abstracting the details away, the $\mr{SIMD}$ operations that can be supported~are addition, element-wise multiplication by a plaintext or ciphertext and data slot permutations that can achieve ciphertext rotations (e.g., used for summing up the values in every slot). In what follows, we will use $+$ and $\odot$ for $\mr{SIMD}$ addition and multiplication and $\rho(\mbf x,i)$ to denote the row vector $\mbf x$ rotated to the left by $i$ positions ($i<0$ means rotation to the right). 
However, once the values are packed, we should try to avoid extractions of individual values from the ciphertexts to improve computation time and circuit multiplicative depth.

In what follows, we denote by $\mr{e_{v0}}(\mbf x)$ the encoding of the vector $\mbf x$ followed by trailing zeros in one plaintext and by $\mr{e_{v\ast}}(\mbf x)$ the encoding of the vector $\mbf x$ followed by junk elements (elements whose value we do not care about, usually obtained through partial operations). We denote by $\mr{e_{vv}}(\mbf x)$ the encoding of the vector $\mbf x$ repeated $\left[ \mbf x \, \mbf x \, \mbf \, \ldots\right]$ in one plaintext. Finally, we also define the encoding of a vector repeated element-wise: for $\mbf x = \left[ \begin{matrix} x_{[1]} & x_{[2]} & \ldots & x_{[n]}\end{matrix}\right]$, this encoding is $\mr{e_{vr0}}(\mbf x) = \left[ x_{[1]} \, \ldots \, x_{[1]} \, x_{[2]} \, \ldots \, x_{[2]} \, \ldots \, x_{[n]} \, \ldots \, x_{[n]} \, 0 \, 0 \, \ldots \right]$, where each element is repeated for a specific number of times. Similarly, we use the notation $\mr{e_{vr\ast}}(\cdot)$ when the repeated elements are followed by junk elements. 
When constructing ciphertexts through encryption, we use similar notations for the encryption of the corresponding encodings of a vector: $\mr{E_{v0}}(\mbf x)$, $\mr{E_{v\ast}}(\mbf x)$, $\mr{E_{vv}}(\mbf x)$, $\mr{E_{vr0}}(\mbf x)$ and $\mr{E_{vr\ast}}(\mbf x)$. 

A scheme has a \textit{security parameter} $\kappa$ if all known attacks against the scheme take $2^\kappa$ bit operations. For CKKS, the security parameter is determined according to the Decisional  Ring  Learning  with  Errors problem hardness~\cite{Albrecht15}. 
We are also interested in the \textit{semantic security} of this encryption scheme, which, at a high level, states that any two ciphertexts are computationally indistinguishable to an adversary that does not have the secret key. 
\section{Co-design of Encrypted Controller}
\label{sec:main_idea}

According to the requirements in Section~\ref{sec:formulation}, the challenges for the encrypted data-driven control can be summarized as:
\begin{itemize}
	\item The computations are iterative and not readily formulated as (low-depth) arithmetic circuits. 
	\item The problem is computationally intensive: it requires large storage, large matrix inversions and many consecutive matrix multiplications.
	\item The precision loss due to the private computations might affect the control performance.
\end{itemize}

To deal with these challenges, we design an encrypted version of the closed-form solution~\eqref{eq:opt_g} of the control~problem stated in Section~\ref{sec:formulation} and manipulate the computations involving matrix inverses to reduce the multiplicative depth needed. We employ the CKKS homomorphic scheme described in Section~\ref{sec:FHE} to address the precision challenge.

First, we approximate problem~\eqref{eq:optimization0} into optimization problem~\eqref{eq:optimization3}, as shown in Section~\ref{sec:reformulation}, to simplify the encrypted computations and to avoid infeasibility due to noise and encryption errors. Second, we aim to write the computation of the solution of~\eqref{eq:optimization3} as a low-depth arithmetic circuit. 
Despite being preferable to an iterative algorithm for solving the optimization problem~\eqref{eq:optimization3} (that would increase the depth at every iteration), the closed-form solution~\eqref{eq:opt_g} involves the encrypted inversion of a matrix, which cannot be generally written as a low-depth arithmetic circuit. 
This is not a problem in the offline feedback version of the problem~\eqref{eq:optimization3}, because this matrix inversion is required only~once and such complex computations can be all performed offline, leaving only three encrypted matrix-vector multiplications to be performed at each time step, as shown in Section~\ref{sec:naSolution}. 
However, in the online feedback version of the problem, this matrix inversion and many consecutive matrix-vector multiplications are required at every time step. 
Nevertheless, we leverage the special structure of the matrix to be inverted by using Schur's complement~\cite[Ch.~0]{ben2003generalized} to reformulate the inverse computation as some lower multiplicative depth matrix-vector multiplications and one scalar division. Furthermore, to avoid performing the division, which is costly on encrypted data, the server sends to the client the denominator such that the client can return its inverse to the server, turning division into multiplication.

We also perform further valuable optimizations in terms of ciphertext packing, adding redundancy in the encoded values, and manipulating the computations to be performed in a $\mr{SIMD}$ mode, which are crucial to the tractability of the solution. 

The multiplicative depth of the arithmetic circuit computing $\mbf u^{\ast,t}_{[0:m-1]}$ can be reduced through reordering of the intermediate operations.
Given a circuit that computes $x$ and a circuit that computes $y$, with multiplicative depth $d(x)$ and $d(y)$, the multiplicative depth of the circuit that evaluates them in parallel and then computes their product is: 
$d({xy}) = \max(d(x),d(y)) + 1$.  
The multiplicative depth should not be confounded with the number of multiplications. Judiciously choosing the order in which to perform the multiplications can reduce the multiplicative depth of the result. For example, consider we want to compute $y = x_1x_2x_3x_4$, where $d({x_i}) = 0$. If we sequentially perform the multiplications, we obtain $d(y) = 3$. However, if we perform $y = (x_1x_2)(x_3x_4)$, we obtain $d(y)=2$. We will use this trick in Section~\ref{sec:solution}. 

Finally, as explained in Section~\ref{sec:FHE}, each ciphertext is created with a number of moduli corresponding to the depth of multiplications it can support before the underlying plaintext is corrupted by noise. After each multiplication and rescaling, a modulus is removed from the ciphertext, leading to more efficient computations. We exploit this fact in our solution, by making sure we compute each operation on ciphertexts that have the minimum required number of moduli. This happens by directly encrypting the message at a lower number of moduli and by removing the number of moduli on the fly. 
\section{Offline feedback encrypted solution}
\label{sec:naSolution}
Recall the closed-form solution of the optimization problem~\eqref{eq:optimization3}. 
We now compute the multiplicative depth of the quantities of interest for consecutive time steps. 
First, the quantities $\bar{\mbf u}_0$, $\bar{\mbf y}_0$, the reference signal and measurements are freshly encrypted at every time step $t$, hence they are inputs to the circuit that computes $\mbf u_t$ and have a multiplicative depth of 1. This means, $ \forall t \geq 0$, 
$d(\mbf y_t) = d(\mbf r_t) = d(\bar{\mbf u}_0) = d(\bar{\mbf y}_0) = 0$. 
Second, we assume that all the quantities obtained offline will have multiplicative depth 0. We do not address the offline computations in this paper; since these computations depend only on the offline data and are one-time, expensive secure solutions can be used, e.g., the cloud could perform the encrypted computations, then perform bootstrapping to refresh the ciphertexts~\cite{Cheon2018bootstrapping,Chen2019improved}. 
Hence, the cloud has fresh encryptions of the following products: $\mbf A_r := \left[\begin{matrix} \mbf I_m & \mbf 0_{(N-1)m} \end{matrix}\right] \mbf U^f \mbf M^{-1}{\mbf Y^f}^\intercal \mbf Q\in\mathbb R^{m\times pN}$, $\mbf A_y := \left[\begin{matrix} \mbf I_m & \mbf 0_{(N-1)m} \end{matrix}\right]\mbf U^f \mbf M^{-1}\lambda_u{\mbf Y^p}^\intercal\in\mathbb R^{m\times pM}$, $\mbf A_u := \left[\begin{matrix} \mbf I_m & \mbf 0_{(N-1)m} \end{matrix}\right]\mbf U^f \mbf M^{-1}\lambda_u{\mbf U^p}^\intercal \in \mathbb R^{m \times mM}$, and the underlying messages are inputs to the circuit and they have multiplicative depth 0. Then:
\begin{equation}~\label{eq:u_offline}
\mbf u_t = \mbf A_r \mbf r_t + \mbf A_y \bar{\mbf y}_t + \mbf A_u\bar{\mbf u}_t.
\end{equation}

For $t=0$, we obtain $d(\mbf u_0) = \max(d(\mbf r_0) + 1, d(\bar{\mbf y}_0) + 1, d(\bar{\mbf u}_0) + 1)=1$. Generalizing:
\begin{equation}\label{eq:non-adaptiveU}
	d(\mbf u_t) = d(\bar{\mbf u}_t) + 1, ~~t\geq 0.
\end{equation}
At time $t+1$, $\bar{\mbf u}_{t+1}$ will be updated by $\mbf u_t$.

We next present three encrypted methods of computing~\eqref{eq:u_offline}, that exhibit trade-offs in depth, memory and communication. 

If we individually encrypt each element of the quantities in~\eqref{eq:u_offline} in a separate ciphertext (the inputs to the circuit computing $\mbf u_t$ will be all elements, rather than three vectors and three matrices), $\bar{\mbf u}_{t+1}$ will have the multiplicative depth of $\mbf u_t$ and $d(\mbf u_t) = t+1$, but memory-wise, a step will require $(m+1)(pN+pM+mM)$ ciphertexts. 

It is more efficient from both storage and computation points of view to encode a vector instead of a scalar in a ciphertext and perform the matrix-vector multiplications by a diagonal method. This vector encoding requires a different analysis. 

\textbf{Diagonal method }~\cite{Juvekar2018gazelle,Alexandru2020private}. We explore a method for efficient matrix-vector multiplication where the diagonals of the matrix are encrypted in separate ciphertexts (the notion of diagonal is extended for rectangular matrices).

\subsubsection{Tall matrix}
\label{appsub:tall}
Consider we want to multiply a tall matrix $\mbf S\in\mbR^{u\times v}$, $u\geq v$, by a vector $\mbf p\in \mbR^v$. To this end, we extract the extended diagonals $\mbf d_i$ of $\mbf S$ such that $\mbf d_i$ will have as many elements as the number of rows $u$, for $i = 0,\ldots,v-1$. The corresponding result $\mbf q = \mbf S \mbf p$ is computed as follows, where $\tilde{\mbf p}$ is $\mbf p$ concatenated with itself as many times such that $\tilde{\mbf p}\in\mbb R^u$:
\[\mbf q = \sum_{i=0}^{v-1} \mbf d_i \odot \rho(\tilde{\mbf p}, i).\]

This way of computing the encrypted matrix-vector multiplication is highly efficient when using ciphertext batching. Specifically, we encrypt each extended diagonal of $\mbf S$ in one ciphertext, which gives us a storage of only $v\leq u$ ciphertexts and we encrypt $\mbf p$ repeatedly in one ciphertext. 
The encrypted result is then computed as:
\begin{equation}
\label{eq:tall}
\mr{E_{v0}}(\mbf q) = \sum_{i=0}^{v-1} \mr{E_{v0}}(\mbf d_i) \odot \rho(\mr{E_{vv}}(\mbf p), i).
\end{equation}

\subsubsection{Wide matrix}
\label{appsub:wide}
In the case of a wide matrix $\mbf S\in\mbR^{u\times v}$ with $u \leq v$, instead of using extended diagonals, we use reduced diagonals. One reduced diagonal $\tilde{\mbf d}_i$, for $i=0,\ldots,v-1$ will have as many elements as the number of rows $u$. The storage is not as efficient since we have to encrypt $v \geq u$ ciphertexts for the diagonals. Apart from this, the computation is identical to~\eqref{eq:tall} and produces $\mr{E_{v0}}(\mbf q)$. 

We can improve the storage (at the cost of some extra rotations) when the number of columns divides the number of rows. In that case, we again extract the extended diagonals $\mbf d_i\in\mbR^v$ for $i=0,\ldots,u-1$. The computation of the matrix-vector multiplication takes the following form:
\begin{align*}
    \mbf z &= \sum_{i=0}^{u-1} \mbf d_i \odot \rho(\mbf p, i )\\
    \mbf q &= \left[\mbf z + \sum_{j=1}^{\lceil\log(v/u)\rceil} \rho\left(\mbf z, u\cdot2^{\lceil\log(v/u)\rceil-j}\right)\right]_{[0:u-1]},
\end{align*}
and produces $\mr{E_{v\ast}}(\mbf q)$ instead of $\mr{E_{v0}}(\mbf q)$.

When $u << v$, it is useful to append rows of zeros such that $u$ divides $v$, and then use the above method.

Let us return to~\eqref{eq:u_offline}. In general, $\mbf A_r,\mbf A_y,\mbf A_u$ are wide matrices, but there can exist particular cases when $m >> p$ or when we want to compute all $\mbf u^{\ast,t}$, leading to tall matrices.

We assume that at the onset of time step $t$, the cloud server has $\mr{E_{v0}}(\bar{\mbf u}_t),\mr{E_{v0}}(\bar{\mbf y}_t), \mr{E_{v0}}(\mbf r_t)$. The cloud server also has $\mr{E_{v0}}(\mr{diag}_i\mbf A_r)$, $\mr{E_{v0}}(\mr{diag}_i\mbf A_y)$, $\mr{E_{v0}}(\mr{diag}_i\mbf A_u)$, i.e., each extended diagonal (whose exact definition varies with the shape of the matrix: tall or wide) of the matrices $\mbf A_r, \mbf A_y, \mbf A_u$ is encrypted in a separate ciphertext. In order to be able to use the diagonal methods, the cloud server first obtains $\mr{E_{vv}}(\bar{\mbf u}_t),\mr{E_{vv}}(\bar{\mbf y}_t), \mr{E_{vv}}(\mbf r_t)$ by rotating and adding $\mr{E_{v0}}(\bar{\mbf u}_t),\mr{E_{v0}}(\bar{\mbf y}_t), \mr{E_{v0}}(\mbf r_t)$. This vector packing substantially reduces memory: $3+\max(m,pN)+\max(m,pM)+mM$ ciphertexts in the worst case and $3+\min(m,pN)+\min(m,pM)+m$ in the best case, depending on the shape of the matrix and the type of diagonal chosen, but will require an extra level. From these quantities, the cloud server computes and sends back to the client one ciphertext containing $\mr{E_{v0}}(\mbf u_t)$ or $\mr{E_{v\ast}}(\mbf u_t)$, depending on the type of diagonals chosen. Assume the cloud server obtained $\mr{E_{v0}}(\mbf u_t)$. 
After that, the cloud server has to create $\mr{E_{vv}}(\bar{\mbf u}_{t+1})$ from $\mr{E_{v0}}(\bar{\mbf u}_{t})$ and $\mr{E_{v0}}(\mbf u_t)$, so it:
\begin{itemize}
    \item rotates $\mr{E_{v0}}(\bar{\mbf u}_{t})$ by $m$ positions to the left;
    \item rotates $\mr{E_{v0}}(\mbf u_t)$ to the right by $(M-1)m$ positions, then adds it to $\rho(\mr{E_{v0}}(\bar{\mbf u}_{t}), m)$. This yields $\mr{E_{v0}}(\bar{\mbf u}_{t+1})+\mr{E_{v0}}([0\,0\,\ldots\,0\,
    (\bar{\mbf u}_t)_{[0:m-1]}])$, 
    the last part appearing because of the rotation to the left;
    \item masks the result in order to truly obtain $\mr{E_{v0}}(\bar{\mbf u}_{t+1})$ (otherwise, the summing and rotation would not produce the correct result);
    \item repeatedly rotates and adds to obtain $\mr{E_{vv}}(\bar{\mbf u}_{t+1})$. 
\end{itemize}  
Because of the CKKS encoding through the Discrete Fourier Transform, (which enables SIMD multiplications, unlike other possible encodings), the masking operation needs to consume a level in order to preserve precision~\cite{Cheon2017homomorphic}. 
Obtaining $\mr{E_{v\ast}}(\mbf u_t)$ instead of $\mr{E_{v0}}(\mbf u_t)$ does not change the analysis, since the same masking applied in the third step also removes the junk elements in $\mr{E_{v\ast}}$. Furthermore, the analysis (not shown here) of other less efficient matrix-vector multiplication methods, such as row and column methods, also require an extra masking.

For $t \geq 1$, equation~\eqref{eq:non-adaptiveU} becomes:
$
	d(\mbf u_t) = d(\bar{\mbf u}_t) + 1= d(\mbf u_{t-1}) + 2 = 2t+1
$.
$\mr{E_{vv}}(\bar{\mbf y}_{t+1})$ can be updated the same way $\mr{E_{vv}}(\bar{\mbf u}_{t+1})$ is updated. However, it is the same cost for the client to encrypt $\mr{E_{v0}}({\mbf y}_{t})$ and $\mr{E_{vv}}(\bar{\mbf y}_{t+1})$, so it can encrypt and send the latter. 
Whenever the allocated multiplicative budget is exhausted, the server can ask the client to send a fresh encryption of $\bar{\mbf u}_t$, at little extra cost. 

Nevertheless, a reasonable and inexpensive option is to ask the client to send along with the encryption of $\bar{\mbf y}_t$~a fresh encryption of $\bar{\mbf u}_t$ at every time step. This implies that $d(\mbf u_t) = 1$, i.e., a multiplicative budget of only~$1$ is required for computing the control input for no matter how many time steps. In constructing the algorithm in Section~\ref{sec:solution}, we prefer this option for the trajectory concatenation phase.
\section{Online feedback encrypted solution}
\label{sec:solution}

\subsection{Computing the solution using arithmetic circuits}
\label{subsec:computation}
Being capable of evaluating polynomials via homomorphic encryption theoretically gives us the possibility of evaluating any function arbitrarily close (using Taylor series for example). In practice, the multiplicative depth and loss of precision of these approximations limit the types of functions we can evaluate. Consequently, division is still prohibitive. Furthermore, large matrix inversion should be judiciously performed, in order to avoid computing products of as many factors as the number of rows.

In this section, we will use the following shorter notation: 
\begin{align*}
&H\mbf U := \left[ \begin{matrix} \mbf U^p \\ \mbf U^f \end{matrix}\right], \,
H\mbf Y := \left[ \begin{matrix} \mbf Y^p \\ \mbf Y^f \end{matrix}\right], \,
h\mbf u := \left[ \begin{matrix} \mbf u^p \\ \mbf u^f \end{matrix}\right], \,
h\mbf y := \left[ \begin{matrix} \mbf y^p \\ \mbf y^f \end{matrix}\right],
\end{align*}
where $\mbf u^{p,f}, \mbf y^{p,f}$ are vectors added at the end of the Hankel matrices $\mbf U^{p,f}, \mbf Y^{p,f}$. We will drop the time subscripts $t$ for conciseness, using prime to denote the next time step.

Let us investigate the inversion of matrix $\mbf M$. 
For step $t=0$, all the values involved in computing $\mbf u^\ast$ in~\eqref{eq:opt_u}, are precollected and can be computed offline. This includes the inverse $\mbf M^{-1}$ and other matrix products. However, at the next time step, cf. line 13 in Algorithm~\ref{alg:online}:
\begin{equation}\label{eq:update}
	{H\mbf U}' := \left[ \begin{matrix}H\mbf U & h\mbf u \end{matrix}\right], \quad
	{H\mbf Y}' := \left[ \begin{matrix}H\mbf Y & h\mbf y \end{matrix}\right]
\end{equation}
The last $m$ elements on the last column of ${H\mbf U}'$, respectively of $h\mbf u$, are the values of $\mbf u_t$ at the previous time step. 

Then, the matrix $\mbf M'\in\mbR^{(S+1)\times (S+1)}$:
\begin{equation}\label{eq:M'}
\mbf M':= {H\mbf Y'}^\intercal \left[\begin{matrix} \lambda_y \mbf I & \\  & \mbf Q   \end{matrix}\right] {H\mbf Y'} + {H\mbf U'}^\intercal \left[\begin{matrix} \lambda_u \mbf I  & \\  & \mbf R \end{matrix}\right] {H\mbf U'} + \lambda_g \mbf I
\end{equation}
is a \textit{rank-1 update} of matrix $\mbf M$. Let
\begin{equation}\label{eq:m}
\mu:= h\mbf y^\intercal \left[\begin{matrix} \lambda_y \mbf I & \\  & \mbf Q   \end{matrix}\right] h\mbf y + h\mbf u^\intercal \left[\begin{matrix} \lambda_u \mbf I  & \\  & \mbf R \end{matrix}\right] h\mbf u + \lambda_g,
\end{equation}
\begin{equation}\label{eq:mbfm}
\mbf m := h\mbf y^\intercal \left[\begin{matrix} \lambda_y \mbf I & \\  & \mbf Q   \end{matrix}\right] H\mbf Y + h\mbf u^\intercal \left[\begin{matrix} \lambda_u \mbf I  & \\  & \mbf R \end{matrix}\right] H\mbf U.
\end{equation}
Specifically, $\mbf M'$ will have the following form:
\begin{equation}\label{eq:rank1}
\mbf M' = \left[ \begin{matrix} 	\mbf M & \mbf m^\intercal \\
						\mbf m & \mu \end{matrix} \right].
\end{equation}

Schur's complement~\cite[Ch.~0]{ben2003generalized} gives an efficient way of computing ${\mbf M'}^{-1}$ from $\mbf M^{-1}$ (assuming $\mbf M^{-1}$ exists), by inverting a scalar $s$ and computing a few matrix-vector multiplications:
\begin{equation}\label{eq:s}
s := \mu - \mbf m {\mbf M}^{-1} \mbf m^\intercal,
\end{equation}
\begin{align}
\begin{split}\label{eq:schur}
	{\mbf M'}^{-1} &= \left[ \begin{matrix} 	\mbf M^{-1} + \frac{1}{s}\mbf M^{-1} \mbf m^\intercal\mbf m \mbf M^{-1} & - \frac{1}{s}\mbf M^{-1} \mbf m^\intercal \\
								-\frac{1}{s}\mbf m\mbf M^{-1} & \frac{1}{s} \end{matrix} \right]. 
\end{split}						
\end{align}

To avoid performing the division $1/s$ on encrypted data, the server will ask the client to perform it on plaintext and send back an encryption of result.

\subsection{Reducing memory and depth of the arithmetic circuit}
\label{subsec:reducing}

In this paper, for simplicity, we consider diagonal cost matrices $\mbf Q$ and $\mbf R$. Otherwise, the multiplication by these matrices requires more complicated encrypted operations.

At a given time, the cloud server has (in an encrypted form) $H\mbf U, H\mbf Y ,\bar{\mbf u}, \bar{\mbf y}$ and $\mbf M^{-1}$, along with the unencrypted costs $\mbf Q, \mbf R$, penalties $\lambda_g,\lambda_u,\lambda_y$ and reference signal $\mbf r$. 
The cloud has to compute the equivalent formulation of~\eqref{eq:opt_u}:
\begin{align}\label{eq:uZ}
\begin{split}
&\mbf u =  \left[\begin{matrix} \mbf 0_{mM} &\mbf I_m & \mbf 0_{(N-1)m} \end{matrix}\right]{ H\mbf U'}\, {\mbf M'}^{-1} \mbf Z,\\
&\mbf Z =  {H\mbf Y'}^\intercal \left[\begin{matrix} \lambda_y \mbf I & \\  & \mbf Q \end{matrix}\right] \left[ \begin{matrix} \bar{\mbf y}\\ \mbf r \end{matrix}\right] + {H\mbf U'}^\intercal \left[\begin{matrix} \lambda_u \mbf I & \\ & \mbf R \end{matrix}\right] \left[ \begin{matrix} \bar{\mbf u}\\ \mbf 0 \end{matrix}\right].
\end{split}
\end{align}

There is a vast number of parameters to tune in the encrypted implementation, such as the packing method from the messages into the plaintexts, the storage redundancy, the order of performing the operations and choice of refreshing some ciphertexts. The trade-offs that these different versions bring are in terms of ciphertext storage, key storage (especially permutation key storage), type and number of operations, precision at the end of the operations and total multiplicative depth of the resulting circuit; these goals are intricately intertwined. 
For example, designing the circuit to have a lower multiplicative depth reduces the size of the ciphertexts and reduces the encryption load at the client, but might involve storing more ciphertexts and performing more computations at the server, compared to a version with a higher multiplicative depth. 
Thus, the main difficulty in making the computations tractable is to astutely batch-encode the vectors and matrices into ciphertexts in order to reduce the memory, and manipulate the operations in order to reduce depth, number of operations and storage. 

We describe here how to change the flow of operations in~\eqref{eq:mbfm}--\eqref{eq:schur} in order to minimize the number of ciphertexts representing the relevant quantities: $H\mbf U$, $H\mbf Y$, $\mbf M^{-1}$, $\mbf m$, $\mu$, $\bar{\mbf u}$, $\bar{\mbf y}$, $\mbf u$, and the number of operations, while keeping the depth to a minimum. We make use of the feature of lattice-based homomorphic encryption schemes of encoding a vector of values into a single plaintext, which is then encrypted in a single ciphertext. 
Since we are also dealing with matrices, we explored several options of how to encode the matrices in ciphertexts: columns, rows, (hybrid) diagonals, vectorized matrix (see~\cite{Alexandru2020private}). Note that the same linear algebraic operation, implemented for different encodings, leads to a different number of stored ciphertexts, multiplicative depth of circuit, and number of SIMD operations. For the problem we tackle in this paper, we found that encoding each column of a matrix into a separate ciphertext minimizes the number of stored ciphertexts and number of operations, while keeping the same depth as when each element is encoded in a different ciphertext and the usual (unencrypted) method of performing the operations is used, as stated in Proposition~\ref{prop:complexity}. However, not all the matrices' columns are encoded in the same fashion, as we will see next. In summary, \textit{having some redundancy inside the ciphertexts, i.e., having some values encoded multiple times, helps optimize the depth and complexity}.

The following proposition summarizes the complexity of the online feedback algorithm, when carrying out the steps in this section. The proof of Proposition~\ref{prop:complexity} is constructive and is given in the remaining of this subsection and Appendix~\ref{app:online}.

\begin{proposition}\label{prop:complexity}
We evaluate the arithmetic circuit for the online feedback algorithm on encrypted data corresponding to Algorithm~\ref{alg:online} for one time step with $O((S+t)^2)$ operations, $O(S+t)$ ciphertexts and $O(S+t)$ rotation keys, at depth $2t+4$, where $S$ is the number of columns of the offline generated block-Hankel matrix and $t$ is the number of online samples accumulated so far.
\end{proposition}

Figure~\ref{fig:inner_vv} shows how to obtain an inner product between two encoded vectors using SIMD multiplications, rotations and additions. Notice that due to rotations, the resulting vector will have the relevant scalar in its first slot, and junk (partial sums) in the following slots. 
Figure~\ref{fig:inner_mv} shows the most efficient way (in terms of required number of operations) to obtain the product between a matrix and a vector when the matrix is encoded as separate columns. Specifically, we need the elements of the vector to be repeated and separately encoded. 

\begin{figure}[ht]
    \centering
    \begin{subfigure}[ht]{0.25\textwidth}
        \includegraphics[width=\textwidth]{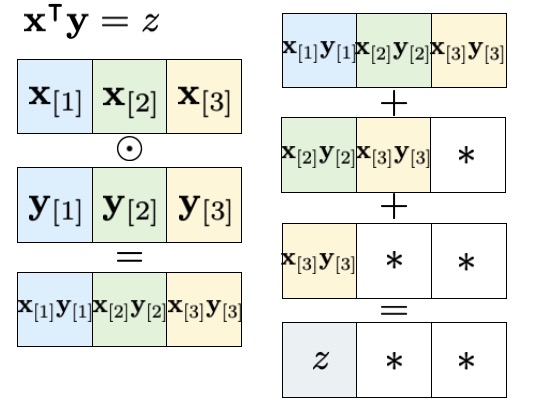}
        \caption{Inner product of two vectors.}
        \label{fig:inner_vv}
    \end{subfigure}
    \quad
    \begin{subfigure}[ht]{0.37\textwidth}
        \includegraphics[width=\textwidth]{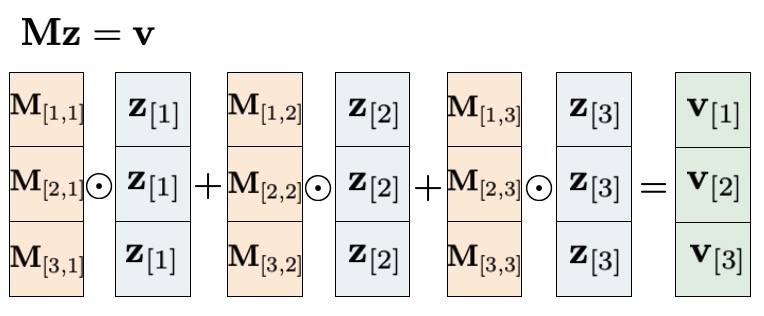}
        \caption{Inner product of a matrix and a vector.}
        \label{fig:inner_mv}
    \end{subfigure}
    ~
    \begin{subfigure}[ht]{0.48\textwidth}
        \includegraphics[width=\textwidth]{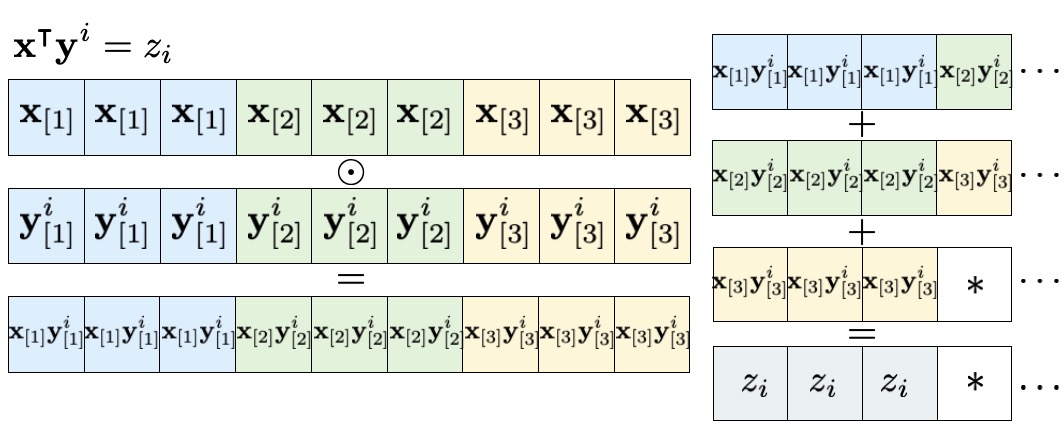}
        \caption{Inner product of two vectors with repeated result.}
        \label{fig:inner_vvr}
    \end{subfigure}    
    \label{fig:inner}
    \caption{Inner product methods for encrypted data.}
\end{figure}

Once we have each column of $\mbf M^{-1}$ encoded in a ciphertext, we want to also obtain the columns of ${\mbf M'}^{-1}$ from~\eqref{eq:schur} encoded each in a ciphertext. To minimize the multiplicative depth and number of operations, this suggests that we need to obtain $\mbf M^{-1}\mbf m^\intercal$ and $\mbf M^{-1}\mbf m^\intercal \mbf m\mbf M^{-1}$ encoded as columns. 

Let us look first at $\mbf M^{-1}\mbf m^\intercal$. In order to use the efficient method outlined in Figure~\ref{fig:inner_mv}, we need to have a separate ciphertext that encodes each element of $\mbf m$, repeated as many times as the number of columns in $\mbf M^{-1}$. In turn, this suggests that we should use the method outlined in Figure~\ref{fig:inner_vvr} when computing $\mbf m$ as in~\eqref{eq:mbfm}, rather than the method outlined in Figure~\ref{fig:inner_vv} which would require extra masking and rotations afterwards.
Specifically, we have to repeat each element of $h\mbf y$ for $S$ times, encode the resulting vector in a ciphertext and do the same for each column of $H\mbf Y$ and the elements in the diagonal matrix $\mbf Q$ and $\lambda_y$. Then, we perform an element-wise multiplication between the ciphertext encoding $h\mbf y$ and the ciphertext encoding the elements of $\mbf Q$ and $\lambda_y$, and then perform the inner product with the encoded columns of $H\mbf Y$. The same steps are taken for $h\mbf u, H\mbf U, \mbf R$ and $\lambda_u$. 

\begin{remark}\label{rem:repeat} 
To account for all the time steps where new samples will be collected, whenever we encode values repeatedly, we need to encode from the beginning $S+\bar T$ copies of the elements, where $S$ is the initial number of columns in $\mbf M$ and $\bar T$ is the total number of new samples intended to be collected.
\end{remark}

One of the main realizations that makes this more efficient computation possible is how to compute each column $\mr{col}_i((\mbf m \mbf M^{-1})^\intercal(\mbf m \mbf M^{-1}))$ at a depth $d(\mbf M^{-1}) + 2$ (when $d(\mbf M^{-1}) \neq 1$), with $O(S)$ storage and $O(S^2)$ operations. Define $\tilde{\mbf M} := (\mbf m \mbf M^{-1})^\intercal(\mbf m \mbf M^{-1})$. Then: 
$\forall i= 0,\ldots,S-1$:
\begin{align}\label{eq:mMMm}
\begin{split}
	\mbf z&:= \sum_{k = 0}^{S-1} \mr{col}_k (\mbf M^{-1}) \odot \left[ \begin{matrix} \mbf m_{[k]} \\ \vdots \\ \mbf m_{[k]}\end{matrix} \right] \\ 
	\mbf w_i &:= \sum_{k=0}^{S-1}  \mr{col}_k(\mbf M^{-1}) \odot \left(\left[ \begin{matrix} \mbf m_{[k]} \\ \vdots \\ \mbf m_{[k]}\end{matrix} \right] \odot \mbf e_i\right),\\
	\mr{col}_i(\tilde{\mbf M}) &= \mbf z \odot \sum_{j=0}^{S-1} \rho \left(\mbf w_i, i-j \right).
\end{split}
\end{align}	

The next piece is to compute $s$ as follows:
\[s = \mu - \sum_{i=0}^{S-1} \left[ \begin{matrix} \mbf m_{[i]} & \ldots & \mbf m_{[i]}\end{matrix} \right]^\intercal \odot \rho(\mbf M^{-1}\mbf m^\intercal, i).\]
The server asks the client to invert $s$. After the server receives from the client $1/s$, it can compute $\frac{1}{s}\tilde{\mbf M}$, as required in~\eqref{eq:schur}, by replacing $\mbf e_i$ in \eqref{eq:mMMm} by $1/s\,\mbf e_i$.

Finally, we construct the columns of ${\mbf M'}^{-1}$ by adding at the end of columns of $\frac{1}{s}\tilde{\mbf M}$ the corresponding extracted elements of $-\frac{1}{s}\mbf M^{-1}\mbf m^\intercal$. The new column that expands the matrix is obtained by adding $1/s$ at the end of the column $-\frac{1}{s}\mbf M^{-1}\mbf m^\intercal$.

Regarding the computation of the control input $\mbf u$, we again need to reorder the operations in order to have the smallest depth possible, i.e. $d({\mbf M'}) + 1$. For that, we also encode $\bar{\mbf u}$ and $\bar{\mbf y}$ with repeated elements, so we can compute the inner product as in Figure~\ref{fig:inner_vvr}, for the elements of $\mbf Z$ in~\eqref{eq:uZ} stored repeatedly. We also prefer some redundancy in storage to avoid online processing, by having a ciphertext for the first $m$ rows of ${\mbf U^f}'$, encoded without repetition. 
This allows us to compute, for $i = 0,\ldots,m-1$:
\begin{align}\label{eq:u_encoded}
\begin{split}
\boldsymbol \upsilon_i &= \sum_{j=0}^S \sum_{k=0}^S \mr{col}_k {\mbf M'}^{-1}  \odot \left ( \mr{row}_i {\mbf U^f}' \odot \mbf Z_{[j]} \right)\\
{\mbf u}_{[i]} &= ({\mbf U^f}' {\mbf M'}^{-1}\mbf Z)_{[i]} = \mbf 1^\intercal \boldsymbol \upsilon_i.
\end{split}
\end{align}
To send only one ciphertext for the result $\mbf u$ (instead of $m$ ciphertexts), while avoiding an extra masking, the server will pack the vectors $\boldsymbol\upsilon_i$ (which have an encoding with trailing zeros, whereas the encoding of ${\mbf u}_{[i]}$ does not) into one ciphertext and ask the client to perform the corresponding summation. 

The cloud service performs the encrypted version of~\eqref{eq:mbfm}--\eqref{eq:schur}, taking the steps outlined in this subsection, which we describe in detail in Appendix~\ref{app:online} in equation~\eqref{eq:encrypted_comp_vector}--\eqref{eq:encrypted_comp_vector3}, along with the specific encoding of each quantity.

\subsection{Precision discussion}
\label{subsec:precision}
\subsubsection*{Scaling}
As described in Section~\ref{sec:FHE} and Appendix~\ref{app:CKKS}, every encrypted operation introduces some noise that corrupts the least significant bits of the encrypted values, which accumulates with the depth of the circuit. This implies that values with a smaller magnitude are affected more by this noise. In the problem we address, the smallest values are found in $\mbf M_t^{-1}$, hence scaling it by a positive integer factor $\alpha$ and accordingly changing the computations such that at every time step we obtain $\alpha \mbf M_t^{-1}$ and $\alpha \mbf u_t$ (such that the client can remove the scaling) allows us to increase the precision of the result. In order to incorporate this scaling without increasing the depth, the server has to perform more operations than before, e.g., recomputation from scratch of some quantities. 
The scaling factor should be selected by the client, and if its magnitude is sensitive, it can be sent as a ciphertext to the cloud server. We can modify the circuit such that this scaling does not affect the depth. 

\subsubsection*{Precision loss when computing the Schur complement}

Given the particularity of the matrix that has to be inverted, at each time step, some precision bits of the Schur complement are lost (depending on the regularization parameter and the magnitude of the measurements), which incurs a loss in precision in the subsequent computations.

Without the regularization term $\lambda_g \mbf I$ added to the objective function of~\eqref{eq:optimization2}, the resulting matrix $\mathbf M$ does not have full rank in the noiseless case (noise helps, but the inverse is not numerically stable). This suggests that the Schur complement~\eqref{eq:s} will have the same order of magnitude as $\lambda_g$, regardless of the values in $\mbf M$, which we show in the following. 
Let $\mbf H^\intercal := \left[ \begin{matrix} H\mbf Y^\intercal & H\mbf U^\intercal \end{matrix} \right]$ and $\mbf h^\intercal := \left[ \begin{matrix}  h\mbf y^\intercal & h\mbf u^\intercal \end{matrix}\right]$ be a new column obtained by adding new samples, as described in Section~\ref{subsec:computation}, and $\mbf P := \mr{blkdiag}( \lambda_y \mbf I, \mbf Q,\lambda_u \mbf I,\mbf R)$. Then,
\begin{align}
    \begin{split}
        s &= \mu - \mbf m \mbf M^{-1} \mbf m^\intercal = \lambda_g + \mbf h^\intercal \mbf P \mbf h - \mbf h^\intercal \mbf P \mbf H (\mbf H^\intercal \mbf P\mbf H + \lambda_g\mbf I)^{-1}\mbf H^\intercal \mbf P \mbf h.
    \end{split}
\end{align}

We now prove that, in the noiseless case, the Schur complement has the same order of magnitude as $\lambda_g$. Formally:
\begin{lemma}\label{lemma:s}
   The following statements characterize the relation between $s$ and $\lambda_g$:
   \begin{enumerate}
       \item[(a)] $\lim\limits_{\lambda_g\rightarrow 0} s = 0$.
       \item[(b)] $\lim\limits_{\lambda_g\rightarrow 0} \frac{s}{\lambda_g} = 1 + \mbf h^\intercal \mbf P^{1/2}(\mbf H^\intercal \mbf P^{1/2})^\dagger (\mbf P^{1/2} \mbf H)^\dagger \mbf P^{1/2}\mbf h$. 
       \item[(c)] $\lim\limits_{\lambda_g\rightarrow \infty} \frac{s}{\lambda_g} = 1$.
       \item[(d)] The function $f(\lambda_g) = \frac{s}{\lambda_g}$ is monotonously decreasing on $[0,~\infty)$.
   \end{enumerate}
\end{lemma}
The proof is given in Appendix~\ref{app:proof_precision}. The term in (b) $\mbf h^\intercal \mbf P^{1/2}(\mbf H^\intercal \mbf P^{1/2})^\dagger (\mbf P^{1/2} \mbf H)^\dagger \mbf P^{1/2}\mbf h$ is small for slowly varying systems. 
Under small random noise, the results in Lemma~\ref{lemma:s} are not exact, but empirically follow closely. 

This analysis is important from an encrypted implementation perspective. Lemma~\ref{lemma:s} shows that, for fixed cost values in $\mbf P$ but regardless of the values in $\mbf H$ and $\mbf h$, i.e., the values of the $\mbf y$ measurements and of the $\mbf u$ control actions, $s$ will be close to the regularization parameter $\lambda_g$. In the encrypted computations, we will compute $\mu = \lambda_g + \mbf h^\intercal \mbf P \mbf h $ and $\mbf m \mbf M^{-1}\mbf m^\intercal$ each with a fixed precision of $x$ bits. For large values of the measurements, we obtain that $s\approx \lambda_g << \mbf h^\intercal \mbf P \mbf h$, which means that there will be a cancellation in the most significant bits (e.g., the first $y$ MSBs), leading to a loss of precision in $s$, which will now have only $x-y$ bits of precision. This is particularly crucial since such a  cancellation of the MSBs will happen at every time step we accumulate new data samples, leading to a cascading loss of precision. Furthermore, due to using the Residue Number System implementation of the encryption scheme, which improves the efficiency by orders of magnitude by using native sizes of 64 bits to store the ``residue ciphertexts" rather than multi-precision arithmetic of arbitrary size, the precision cannot be increased indefinitely (see~\cite{Cheon2018full,Halevi2019improved} for technical details). 

As a side note, decreasing or increasing the values in the cost matrix $\mbf P$ independently from $\lambda_g$ does not solve this precision issue, because it implies changing the problem (it is equivalent to increasing or decreasing $\lambda_g$).

The conclusion of this discussion is that, apart from the role it has for regularization and noise robustness, the parameter $\lambda_g$ also affects the precision of the solution (meaning the difference between the encrypted solution and the unencrypted solution), which suggests we need to pick $\lambda_g$ to not be too small. 
This introduces a precision versus accurate convergence trade-off. We will comment more on this trade-off in Section~\ref{subsec:78bits}.

\subsection{Considerations for continuous running}
Accumulating new samples serves to robustify the algorithm and adapt to new disturbances. However, an important caveat is that adding too many new samples damages the performance of the algorithm, both because of overfitting and because of the problem becoming intractable as the number of variables in~\eqref{eq:optimization3} grows. In future work, we will investigate this issue more, along with having a sliding window of sample collection. 

Nevertheless, there are cases where the number of new samples to be collected leads to a circuit of a depth higher than the preselected multiplication budget. Our options for continuing the computations are:

(i) Restore the initial precollected Hankel matrices which bypasses the refreshing step altogether. 

\noindent Advantages: no extra computations needed. Disadvantages: this causes oscillations in the control actions. 

(ii) Stop adding new information to the matrices after the multiplication budget is exhausted. 

\noindent Advantages: no extra computations needed. Disadvantages: the multiplicative budget has to be large enough such that \textit{enough} samples are collected. 

(iii) Pack the matrix into a single ciphertext as described in~\eqref{eq:packingM_1} and ask the client to refresh it. 
{ \medmuskip=0mu\thickmuskip=1mu
\begin{equation}\label{eq:packingM_1}
    \mr{E_{v0}}(\mbf M^{-1}_t) = \sum_{i=0}^{S+t-1} \rho\left ( \mr{E_{v0}}(\mr{col}_i\mbf M^{-1}_t), -i(S+t-1)\right).
\end{equation}
}
\noindent Advantages: the server can continue collecting values for any desired time, without extra multiplication depth. Disadvantages: the client has to decrypt, encrypt and send another ciphertext; the rotation keys necessary for packing can occupy a lot of storage (but compared to~\cite{Alexandru2020towards}, here we only need $S+t$ such rotations, not $(S+t)^2/2$).

(iv) Bootstrap the ciphertext of $\mbf M^{-1}_t$. The computation advancements~\cite{Chen2019improved} regarding the bootstrapping procedure suggest that it is likely to locally resolve the refreshing step. 

\noindent Advantages: the server can continue collecting values for any desired time without the client's intervention. Disadvantages: the initial multiplication budget has to be larger to also allow for the bootstrapping circuit.

In the solution we implemented, we chose option (iii). This gives us flexibility on the maximum multiplicative depth of the circuit, which will now be $2t_{refresh} + 4$, at little extra cost for the client. 

We can pack as many values as half the ring dimension $\mr{ringDim}$ in one ciphertext.  
Then, the number of ciphertexts the server can pack $\mbf M^{-1}_t$ into is $ \lceil(S+t)(S+t)/(\mr{ringDim}/2) \rceil$. This is viable since the ring dimension is in general large (e.g. $\geq 4092$), in order to accommodate a reasonable multiplicative budget, plaintext precision and standard security parameter. The client only has to decrypt, re-encrypt and send back this number of ciphertexts. The server then uses one extra level to perform the reverse of~\eqref{eq:packingM_1} to unpack $\mr{E_{v0}}(\mbf M^{-1}_t)$ into ciphertexts $\mr{E_{v0}}(\mr{col}_i(\mbf M^{-1}_t))$. These multiplications can be absorbed in the same initial multiplicative depth.

\section{Privacy of the encrypted algorithms}
\label{sec:privacy}

\begin{definition}\label{def:client_privacy}(Client privacy w.r.t. \mbox{semi-honest} behavior of server) 
Let $x_C$ be the private input of a client and $x_S$ be the input of a server. The client wants the server to evaluate a functionality $f$ and return the result $f(x_C,x_S)$. We will denote by $f_S(x_C,x_S)$ the corresponding result at the server. Let $\Pi$ be a \mbox{two-party} protocol for computing $f$. The \textbf{view} of the server during an execution of $\Pi$ on the inputs $(x_C,x_S)$, denoted $V^\Pi_S (x_C,x_S)$, is $(x_S,r,\mr{msg})$, where $r$ represents the outcome of the server's internal coin tosses, and $\mr{msg}$ represents all the messages it has received from the client. 
For a deterministic functionality $f$, we say that $\Pi$ \textbf{privately computes} $f$ w.r.t. to the server if there exists a probabilistic polynomial-time algorithm, called \textbf{simulator} and denoted $\mr{S_S}$, such that, for any inputs $x_C,x_S$:
 \[\{\mc{S}_S(x_S,f_S(x_C,x_S))\}_{x_C,x_S\in\{0,1\}^\ast} \stackrel{c}{\equiv} \{V^\Pi_S (x_C,x_S)\}_{x_C,x_S\in\{0,1\}^\ast},\]
where $\stackrel{c}{\equiv}$ denotes computational indistinguishability
\end{definition}

It is common for a server to be only employed for computing a result that is then sent to the client, and the server has no output formally, i.e., $f_S(x_C,x_S) = \emptyset$.

In our case, the functionality $f$ in Definition~\ref{def:client_privacy} represents the data-driven algorithm, and $f_S(x_C,x_S)$ represents the output of the server after evaluating the functionality on the given inputs. 
Fix a number of time steps $\mc T$. We want to prove a secure evaluation of the data-driven control functionality for $\mc T$ steps, after which we assume the protocol ends. In reality, because the security of the encryption scheme used is based on computational problems, after a long time (years), it is recommended that the secret key is changed; so $\mc T$ represents a natural stopping point and is not restrictive.

Under the aforementioned condition, the output of the protocol at the server, $f_S(x_C,x_S)$, is actually the empty set. 
Furthermore, all the intermediate messages in the view of the server will be ciphertexts encrypted with the client's public key using the CKKS scheme. Finally, the ring dimension of the CKKS scheme is selected such that, for the multiplicative budget of the desired functionality, an adversary cannot brute force the encryption. These are the main observations that will be used in proving the subsequent theorems. 

There is a preprocessing protocol, where the server gets and computes the encrypted information it will have to use for the online control of the system, i.e., the encrypted data trajectories, costs and parameters. This preprocessing protocol is proven secure under the same Definition~\ref{def:client_privacy}, hence, there is an ideal functionality $\mc F^{pre}$, which is essentially indistinguishable from the real-world functionality of the preprocessing protocol. We work in a $\mc F^{pre}$-hybrid model, which means that we can securely compose this functionality with the functionality of the online protocols (see~\cite{Lindell17} for more details about simulation proofs in the hybrid model and composition theorems).

\begin{theorem}
\label{thm:nas}
The encrypted offline data-driven control algorithm in Section~\ref{sec:naSolution} achieves privacy with respect to the server.
\end{theorem}

\begin{proof}
At time $t\leq \mc T$, the client's input is $H\mbf U$, $H\mbf Y$, $\mbf u_{0:t-1}$, $\mbf y_{0:t-1}$, $\mbf r_{0:t-1}$ and $\mbf{pk}$, where $\mbf{pk}$ is the public key of the leveled homomorphic scheme used. We assume that, after the preprocessing protocol, the server has input $\mbf Q$, $\mbf R$, $\lambda_g$, $\lambda_u$, $\lambda_y$, $\mbf{pk}$ and receives as messages $\mr{E}_{\mr{v0}}(\mr{diag}_i\mbf A_r)$, $\mr{E}_{\mr{v0}}(\mr{diag}_i\mbf A_y)$, $\mr{E}_{\mr{v0}}(\mr{diag}_i\mbf A_u)$. At time $t$, the server receives from the client either  $\mr{E}_{\mr{vv}}(\bar{\mbf u}_t),\mr{E}_{\mr{vv}}(\bar{\mbf y}_t), \mr{E}_{\mr{vv}}(\mbf r_t)$ or $\mr{E}_{\mr{vv}}(\bar{\mbf y}_t), \mr{E}_{\mr{vv}}(\mbf r_t)$. After applying the functionality as described in Section~\ref{sec:naSolution}, the server obtains $\mr{E}_{\mr{v0}}(\mbf u_t)$ or $\mr{E}_{\mr{v\ast}}(\mbf u_t)$ and sends it to the client.

The proof immediately follows from the fact that the server only receives fresh ciphertexts from the client. Since the homomorphic encryption scheme is semantically secure, these ciphertexts are computationally indistinguishable from random encryptions, so we can construct a simulator for the server that replaces the true messages from the client by random encryptions of the same size. The input-output distribution of such a simulator will be indistinguishable from the input-output distribution of the true protocol.
\end{proof}

\begin{theorem}
\label{thm:adaptive}
The encrypted online data-driven control algorithm in Section~\ref{sec:solution} achieves privacy with respect to the server.
\end{theorem}

\begin{proof}
The input of the client is the same as in the proof of Theorem~\ref{thm:nas}. The input of the server is $\mbf Q, \mbf R, \lambda_g, \lambda_u, \lambda_y, \mbf{pk}$ and after the preprocessing protocol, it has the corresponding ciphertexts for $H\mbf U_0,H\mbf Y_0,\mbf M_0^{-1}, \mbf U^f_0$. At time $t<\bar T$, while collecting new samples, the server has ciphertexts corresponding to $H\mbf U_{t-1},H\mbf Y_{t-1},\mbf M_{t-1}^{-1}, \bar{\mbf u}_{t-1}, \bar{\mbf y}_{t-1}$, receives from the client $\mr{E}_{\mr{vr0}}(\mbf u_t),\mr{E}_{\mr{vr0}}(\mbf y_t), \mr{E}_{\mr{vr0}}(\mbf r_t)$, then computes $\mr{E}_{\mr{vr\ast}}(s_t)$ and receives from the client $\mr{E}_{\mr{vr0}}(1/s_t)$. The analysis in Section~\ref{subsec:precision} does not reveal private information: although the cloud server knows $\lambda_g$, it never gets $s_t$ in cleartext. Finally, after applying the functionality described in Section~\ref{sec:solution} and Section~\ref{app:online}, the server obtains the corresponding ciphertexts of $H\mbf U_{t},H\mbf Y_{t},\mbf M_{t}^{-1},\mbf U^f_t, \bar{\mbf u}_{t}, \bar{\mbf y}_{t}$ and $\mr{E}_{\mr{v\ast}}(\boldsymbol \upsilon_t)$, and sends the latter to the client, which decrypts it and computes $\mbf u_t$. During the times designated for refreshing $\mbf M_t^{-1}$, the server packs it into a single ciphertext, sends it to the client and receives $\mr{E}_{\mr{v0}}(\mbf M_t^{-1})$, which uses new uncorrelated randomness. After the collection of the new samples ends, i.e., $\bar T\leq t \leq \mc T$, the server only updates the ciphertexts corresponding to the quantities $\bar{\mbf u}_t$ and $\bar{\mbf y}_t$, but otherwise computes the same output. 

Since the server only receives fresh encryptions of the private data, we can construct a simulator for the server that replaces the true messages from the client by random encryption of corresponding size and apply the server's functionality. Privacy follows because of the semantic security of the underlying CKKS scheme.
\end{proof}

\section{Extensions}
\label{sec:extensions}
We now briefly describe extensions and possible modifications to the online algorithm.

\subsection{Inputting blocks of control inputs}
For simplicity of exposition, we illustrated the encrypted computations outline for when the server sends only the first component of $\mbf u^{\ast,t}$ to the client. In practice, it is common that more consecutive components from the $N$ computed control inputs are applied. This is also beneficial for the encrypted solution, since the most expensive computations, the inversion of $\mbf M_t$ and its subsequent update, need to be computed only once for multiple time steps. Similarly, the communication rounds between the client and the server would be reduced. We recommend this for larger systems.

\subsection{Batch collection of new samples}
We can still use the Schur complement to compute the inverse of a rank-$x$ update of the matrix $\mbf M_t$, where $x$ is how many new data samples we want to accumulate at once. The server then packs the matrix representing the Schur complement into one ciphertext and asks the client to invert it and send it back. The circuit depth will increase accordingly. 

\subsection{Sliding window of sample collection}

In order to manage the growth in the number of variables, i.e., number of columns of $H\mbf U_t$ and $H\mbf Y_t$, after collecting a sample and updating $\mbf M_t^{-1}$, we would like to remove the first sample collected. We can again achieve this by using Schur's complement and the Woodbury matrix inversion lemma. We revert to the notation used in Section~\ref{sec:main_idea} for simplicity of the exposition and we show how to remove the first sample, i.e., the first row and first column of $\mbf M'$ (for removing rows and columns inside the matrix, we will have to multiply by permutation matrices). 
\begin{align*}
\mbf M' &= \left[ \begin{matrix} 	\mbf M & \mbf m^\intercal \\
						\mbf m & \mu \end{matrix} \right]
		=: \left[ \begin{matrix} 	\nu & \mbf n \\
						\mbf n^\intercal & \mbf N \end{matrix} \right].
\end{align*}
Given ${\mbf M'}^{-1}$, we would like to extract $\mbf N^{-1}$. 
The Schur complements expressions for the two matrices are:
\begin{align*}
    \mbf M'/\mbf M &:= \mu - \mbf m \mbf M^{-1}\mbf m^\intercal = m_{S} , \quad \mbf M'/\mbf N := \nu - \mbf n \mbf N^{-1}\mbf n^\intercal, \quad \mbf M'/\nu := \mbf N - \mbf n^\intercal \nu^{-1}\mbf n.
\end{align*}

This gives the expression in~\eqref{eq:schur} and the following equal expression for $(\mbf M')^{-1}$:
\begin{align*}
	{\mbf M'}^{-1} &= \left[ \begin{matrix} 	(\mbf M'/\mbf N)^{-1} & - (\mbf M'/\mbf N)^{-1} \mbf n\mbf N^{-1} \\
								- \mbf N^{-1}\mbf n^\intercal(\mbf M'/\mbf N)^{-1} & (\mbf M'/\nu)^{-1} \end{matrix} \right] =: \left[\begin{matrix} l_1 & \mbf L_2 \\ \mbf L_2^\intercal & \mbf L_3. \end{matrix}\right].
\end{align*}
By the matrix inversion lemma,
\begin{align*}
    \mbf L_3 &= \mbf N^{-1} + \mbf N^{-1}\mbf n^\intercal(\mbf M'/\mbf N)^{-1} (\mbf M'/\mbf N)(\mbf M'/\mbf N)^{-1}\mbf n\mbf N^{-1}= \mbf N^{-1} + \mbf L_2^\intercal l_1^{-1} \mbf L_2\\
    \mbf N^{-1} &= \mbf L_3 - \mbf L_2^\intercal l_1^{-1} \mbf L_2.
\end{align*}

Removing a sample after collecting one sample incurs an increase by two in the total depth of the computation and asking the client to perform the division of $l_1$. 

\subsection{Reducing accumulated noise}

As discussed in Section~\ref{subsec:precision}, noise accumulates in the stored ciphertexts of $\mbf M_t^{-1}$ at the server, because of the fixed precision imposed by the native dimension of 64 bits (in the Residue Number System implementation of the encryption scheme).
However, before the noise grows too much, the server could try to keep track of it and reduce it by using the fact it can also easily compute $\mbf M_t$ as in~\eqref{eq:M'}. Denote the error on $\mbf M_t^{-1}$ by $\mbf X$. Then:
\begin{align}\label{eq:error_M_1}
\begin{split}
    \boldsymbol \Lambda := \mbf M_t (\mbf M^{-1}_t + \mbf X) &= \mbf I + \mbf M_t \mbf X\\
    (\mbf M^{-1}_t + \mbf X)(\boldsymbol \Lambda - \mbf I) &= \mbf X - \mbf X \mbf M_t \mbf X.
\end{split}
\end{align}

It is safe to assume $|| \mbf X ||_\infty \leq \epsilon ||\mbf M_t||_\infty$ (other norms can be chosen), with $\epsilon < 1$. Furthermore, if $\epsilon$ satisfies the relation $\epsilon || \mbf M_t ||^2_\infty \leq 1$, then $||\mbf X\mbf M_t \mbf X||_\infty \leq ||\mbf X||_\infty$. This would mean that the server can compute a better approximation of $\mbf M_t^{-1}$ than $\mbf M^{-1}_t + \mbf X$ as $\mbf M^{-1}_t + \mbf X\mbf M_t\mbf X$. However, this second approximation would imply a larger depth for $\mbf M^{-1}_t$, so the best option is for the server to compute this approximation right before it asks the client to refresh the ciphertext.

\subsection{Function privacy}\label{sec:function_privacy}

The algorithm run at the CaaS service provider might be proprietary. In that case, the client should not get extra information about that algorithm, apart from what it agreed to receive. This notion is captured by the \textit{function privacy} property of a homomorphic encryption scheme, which states that a ciphertext should not expose information about the function that was evaluated in order to create that ciphertext. Inherently, the CKKS scheme does not preserve function privacy because of the partial information stored in the junk elements. The common ways to achieve function privacy is for the server to mask out the irrelevant slots before handing a ciphertext to the client to decrypt, either by a multiplicative mask (zeroing out the junk elements) or by adding large enough randomness to them. The multiplicative mask increases the depth by one, whereas the randomness masking can lead to an increase in the plaintext modulus since the randomness has to be large enough to drown the relevant information. 

In our scenario, the client is allowed to obtain the numerator and denominator of the control input at every time step and also the inverse matrix $\mbf M^{-1}$, which is only determined from the inputs and outputs that the client knows. The server can zero out the extra slots in $\mr{E_{vr\ast}}(s_t)$ and $ \mr{E_{v\ast}}(\boldsymbol \upsilon_t)$. This masking would add one extra level to the respective ciphertext, but the level increase would not build up, since these masks are not used in posterior computations. Further investigation on keeping the service provider's algorithm and costs private from the client will be subject to future work.
\section{Implementation and evaluation}
\label{sec:simulations}
We focused on testing out our proposed algorithms on a temperature control problem, because first, smart buildings can support flexible encryption methods for the collected data, and second, the sampling time is of the order of minutes, which allows for complex computations and communication between the client and the server to take place.

\subsection{Comparison with algorithm in~\cite{Alexandru2020towards}}
We first want to underline the improvements in both runtime and memory that the algorithm presented in the current work brings compared to the algorithm in~\cite{Alexandru2020towards}. For the same parameters considered in the zone temperature problem exemplified in~\cite{Alexandru2020towards}, which had one input and one output, used a ring dimension of $2^{12}$ and 21 levels in the ciphertext, we used the same commodity laptop with Intel Core i7, 8 GB of RAM and 8 virtual cores at 1.88 GHz frequency. We obtained an improvement of 2.4x in terms of memory: 2.25 GB compared to 5.48 GB. We obtained an improvement of 4x in terms of total running time, and specifically, the maximum runtime per step dropped from 308.8 s (previous algorithm) to 75 s (current algorithm). Some of the changes (e.g., in terms of better thread parallelization) we applied to our current algorithm can also be applied to the algorithm in~\cite{Alexandru2020towards}. These improvements allowed us to simulate the encrypted control for larger systems with better security parameters, as described in the next parts.

\subsection{2x2 system with 73 bits of security}
\label{subsec:78bits}

\subsubsection{Setup}
We considered a temperature control problem for the ground floor of a building with two zones. The~sampling time is of $T_s = 420$ s. The system parameters are: $n = 2, m = 2, p = 2, M = 4, N = 4, T = 32$ and the system is stable. 
To mimic a realistic example, we considered two types of disturbances: unknown, characterized by a zero mean Gaussian process noise with covariance $0.001\mbf I$ and a zero mean Gaussian measurement noise with covariance $0.01\mbf I$, and known, slowly varying disturbances caused by the exterior temperature and the ground temperature, which we simulate as varying uniformly at random between $[24.5^\circ, 25.5^\circ]$ and $[9.5^\circ, 10.5^\circ]$, respectively. We choose the cost matrices and regularization terms $\mbf Q = \mbf I, \mbf R = 10^{-3}\mbf I, \lambda_g = 5, \lambda_y = \lambda_u = 1$. 

For the offline data collection, we assume a different initial point than for the online computation, and randomly sample the offline input trajectory so that the corresponding output trajectory lies in the interval $[10^\circ,40^\circ]$.
The $M,N$ and $T$ parameters mean that we start from 25 columns in the Hankel matrices. We note that, due to noise, the offline feedback control algorithm has suboptimal performance for these parameters--see Figure~\ref{fig:comparison}. 
Thus, we collect data online for 15 more time steps, which means adding another 15 columns to the Hankel matrices. Afterwards, we run the system with fixed gains. Figure~\ref{fig:comparison} shows the performance of the setpoint tracking with these parameters and under noise and disturbances.

\begin{figure}[ht]
	\centering
	\includegraphics[width=0.55\columnwidth]{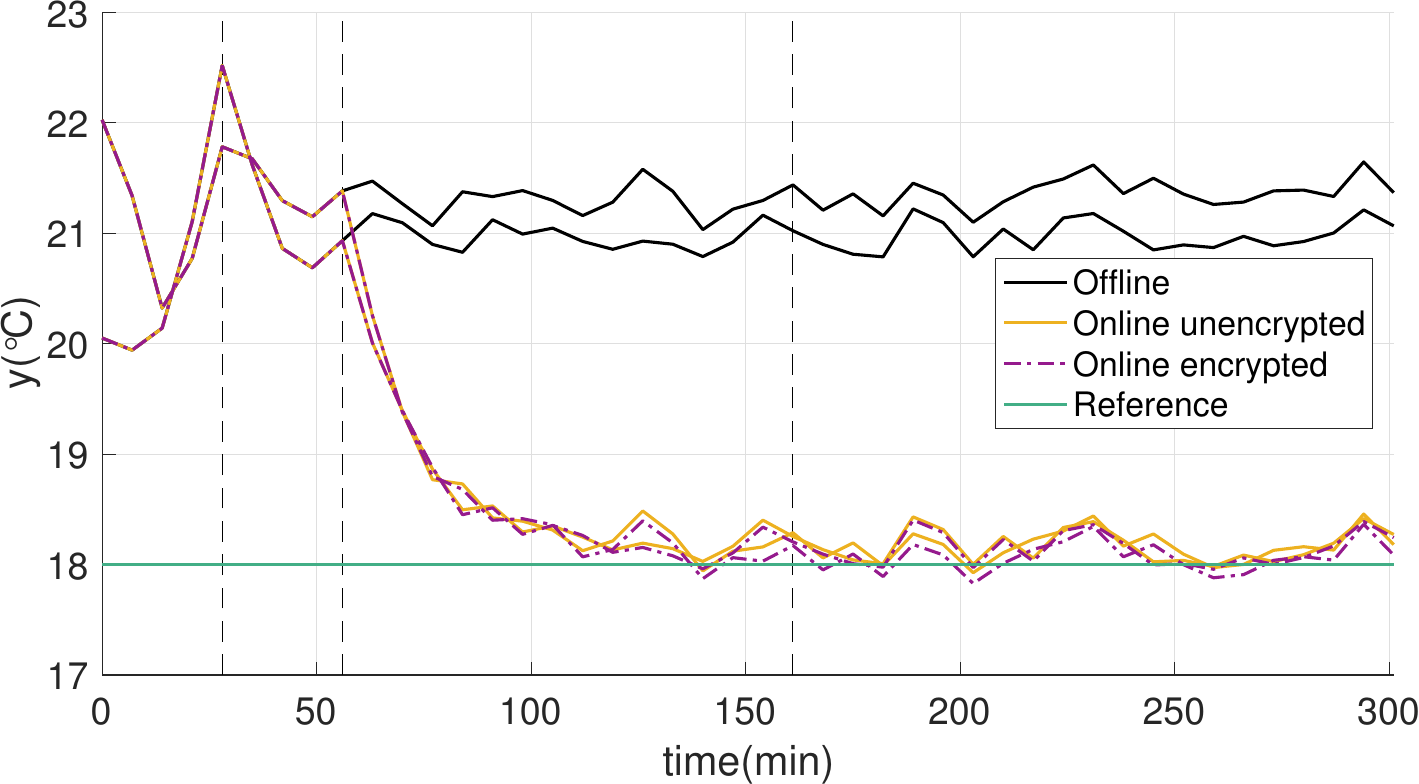}
	\caption{Tracking performance of the online versus offline feedback in the presence of noise. To make the comparison meaningful, we use the same noise sequence for both schemes, in both the encrypted and unencrypted runs of the algorithms. The first vertical dashed line marks the first $M$ time steps, corresponding to the initial offline data, the second vertical dashed line marks the following $N$ time steps, corresponding to the trajectory concatenation, and the last vertical dashed line marks the end of sample collection.}
	\label{fig:comparison}
\end{figure}

\subsubsection{Implementation details}
For better benchmarking capabilities (less variability between runs), for these experiments we used an AWS EC2 c5.2xlarge machine, with 8 virtual cores, 3.4GHz frequency and 16 GB of RAM. 
We implemented\footnote{https://github.com/andreea-alexandru/swift-private-data-driven-control} the encrypted solutions~proposed in this paper using the PALISADE library \cite{PALISADE}. 

We use a ring dimension of $2^{14}$ and 15 levels (for a refresh step after 5 collected samples).  
The resulting ciphertext modulus of a fresh ciphertext is of 760 bits (the first plaintext modulus has 60 bits and the following plaintext moduli have 50 bits). This gives a security size of 73 bits, according to the LWE estimator~\cite{Albrecht15}. Choosing a smaller refreshing time increases the communication rounds between the client and the server (one extra exchange at every 35 minutes in this case), but at the same time decreases the multiplicative depth of the circuit, reducing the computational complexity and the size of the parameters. 

\subsubsection{Precision}
Figure~\ref{fig:comparison} shows the tracking performance for the encrypted and unencrypted versions of the online and offline feedback algorithms. Due to the precision loss incurred by working with a deep circuit on encrypted data, we see a small offset in the measured output compared to the unecrypted measured output for the online feedback algorithm. Since the offline feedback algorithm  is a very shallow circuit (depth 1) when the client returns a fresh encryption of the control input to the server, it incurs no precision loss. 

We obtain a maximum magnitude of  $0.129^\circ$C and an average of $0.045^\circ$C for the difference between the~unencrypted measured output and the encrypted output during the online simulation of a noisy process. For the control input, the maximum difference is $0.248$ kWatts and an average difference of $0.053$ kWatts. A typical such sample is depicted in Figure~\ref{fig:comparison}. Notice that the maximum magnitude of the difference is smaller than the variation in the disturbances and does not affect the performance and stability of the control algorithm. One can increase the plaintext modulus in order to reduce the error introduced by encryption (e.g. 53 bits instead of 50 bits in the moduli will lead to an almost~tenfold decrease in the differences). However, if we keep the ring dimension constant, this reduces the security level to 69 bits. In Section~\ref{subsec:more_security}, we show results for better precision at a standard security level. As described in Section~\ref{subsec:precision}, the loss of precision can be alleviated by a larger regularization parameter $\lambda_g$. For example, keeping all other parameters the same but increasing $\lambda_g = 10$ will lead to a maximum magnitude of the difference in the measurements of $0.037^\circ$ and the maximum magnitude of the difference in the inputs of $0.043$ kWatts, but the tracking performance is slightly reduced (from $0.06^\circ$C to $0.12^\circ$C deviation from the reference in the noiseless case).

\subsubsection{Offline}
The encrypted offline feedback algorithm is very fast, since it has only depth 1 when the client sends to the server encryptions of $\bar {\mbf u}_t$ and $\bar {\mbf y}_t$. A ring dimension of $2^{12}$ gives a security parameter of 126 bits, with a ciphertext having a size of 110 KB. The peak online RAM is 42 MB. The average runtime for a time step is 57 ms for the server and 12 ms for the client. 

\subsubsection{Online} We now examine the encrypted online control algorithm simulation adding details that complement the information in Section~\ref{sec:solution}.
\subsubsection*{Memory}
The peak RAM for both the client and the server incurred during the three phases for 45 time steps was 6.26 GB, out of which, 3.9 MB is the public key, 2 MB is the secret key, 15.7 MB is the relinearization key and 4.6 GB are the evaluation rotation keys (corresponding to a key for each of the 336 rotation indices), generated offline. The client discards the evaluation keys after it sends them to the server. At every iteration, the client receives two ciphertexts from the server (the Schur complement and control action) and sends three ciphertexts (inverse of the Schur complement, measurement, control action--we assumed here that the setpoint is constant). During the refresh time steps, one extra ciphertext is sent and received (the inverse matrix). The size of these ciphertexts depends on the level where the ciphertext is at. Specifically, for our selected parameters described in the beginning of the subsection, the ciphertexts are 2.77 MB at maximum size and 240 KB at the minimum size.

\subsubsection*{Time}
The runtimes for the encrypted computation are given in Figure~\ref{fig:running_time}, where three different phases are depicted: trajectory concatenation, online sample collection and static update. The total offline initialization time for key generation and ciphertext encryption is 60 s for both the client and the server. In all online phases, the client only performs cheap and fast computations that take less that 0.3~s. 

\begin{figure}[ht]
	\centering
	\includegraphics[width=0.5\columnwidth]{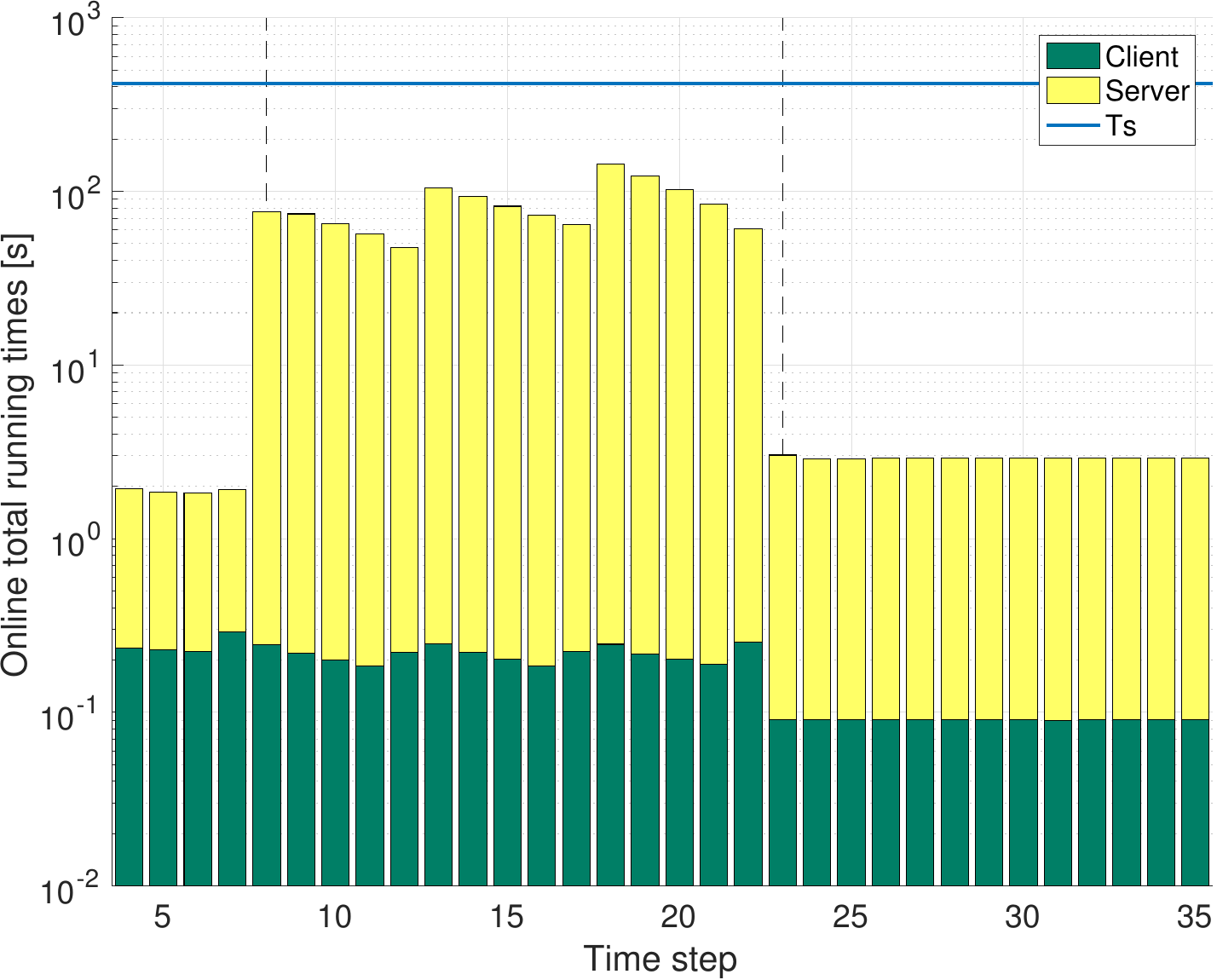}
	\caption{Running times for the computations performed at the client and the cloud server for the encrypted online control algorithm. The plot is semi-logarithmic and the amounts of time required by the client and the cloud server are stacked. 
	The first $M$ time steps are not depicted, because the server has no computational load. The first vertical dashed line marks the following $N$ time steps, corresponding to the trajectory concatenation, and the last vertical dashed line marks the end of sample collection.}
	\label{fig:running_time}
\end{figure}

The first online phase is the trajectory concatenation phase, as described in Section~\ref{subsec:online}, from $t = M-1 = 3$ to $t = M + N - 1 = 7$, where the server computes the encrypted control action only with the precollected data. This is very efficient, despite computing on the maximum ciphertext sizes, and results in a total computation time for one time step of 2~s. Before $t = M-1$, the client applies random inputs. 

The second phase is the online collection of new input-output samples, which implies modifying the Hankel matrices and computing the inverse matrix via the Schur complement at every time step $\mbf M^{-1}_t$ from step 8 to step 22. This phase is split in as many parts as  the refreshing time dictates. In the simulation depicted in Figure~\ref{fig:running_time}, this corresponds to three parts. At the established refresh times ($t = 12, 17$), the server packs the matrix $\mbf M^{-1}_t$ into one ciphertext, sends it to the client to re-encrypt it with the maximum number of moduli, and unpacks it back into column-wise ciphertexts. First, the increase in runtime between the subsequent time steps 13 and 12, respectively 18 and 17, is given by the fact the ciphertexts are returned to the maximum number of moduli. Second, the increase in computation time from the first refresh at time 12 and 13 (47.23 s and 104.6 s) to the second refresh at time 17 and 18 (63.9 s and 144.2 s) is given by the fact that the server has to deal with more collected samples ciphertexts than in the beginning of phase two. The intermediate decrease in the computation time is given by the decrease in the ciphertext size (we make sure to compute with only the minimal number of moduli required). 

The third phase corresponds to the computations after stopping the collection of new samples, which starts from time step 23 and can go for the rest of the desired simulation time (accounting for Remark~\ref{rem:ringv0_vector} in Appendix~\ref{app:online}). The third phase is only slightly more computationally intensive than the first. The time for the client halves compared to the previous two phases, because the client has to decrypt and encrypt ciphertexts with two moduli. The running time for the server substantially decreases compared to phase two but doubles compared to phase one, because we use a less efficient matrix-vector multiplications in order to minimize the multiplicative depth. Nevertheless, the running time required for computing the control input at one time step is around 3~s. 

\subsection{More security}
\label{subsec:more_security}

In practice, the desired security parameter for a cryptographic application is at least 100 bits. The refresh time after 5 time steps (35 minutes) gives a ciphertext modulus of 802 bits and the refresh time after 8 time steps (56 minutes), which means a deeper circuit, gives a ciphertext modulus of 1120 bits. In Table~\ref{tab:simulations}, we present simulation results where we choose the ring dimension $2^{15}$ such that we satisfy this security requirement for the chosen ciphertext modulus. We simulate for the above 2x2 system and for a 4x4 system (where $M = 4, N = 4, T = 64$ and we again collect 15 more online samples, which gives a maximum number of 72 unknowns). We use 53 bits of precision for the plaintext moduli, such that we obtain $0.034^\circ$C and $0.008^\circ$C maximum, respectively average difference between the encrypted and unencrypted measurements, and $0.05$ kWatts and $0.008$ kWatts maximum, respectively average difference for the input.  

We use different AWS cloud machines in order to satisfy the RAM and sampling time requirements. 
The c5.2xlarge machine (16 GB RAM) used 8 threads, the c5.4xlarge machine (32 GB RAM ) used 13 threads and the c5.9xlarge machine (72 GB RAM) used 28 threads. 

From Table~\ref{tab:simulations}, we see the trade-off between a longer refresh time (the server asks the client to refresh a ciphertext after more time, which implies a deeper circuit and less security for the same ring dimension) and the memory consumption and computation time. The conclusion is that \textit{the refresh time is a tuning knob between security level, communication and computational complexity which should be determined according to the application's specifications.}

Note that as the performance of the underlying homomorphic encryption library improves, these improvements will reflect in our computation times as well.

\begin{table*}[ht]
\centering
\caption{
Simulation results of the encrypted online feedback algorithm for systems of two different sizes, at different security levels and different refresh times. In all entries, the ring dimension is $2^{15}$, to ensure a security level of at least 100 bits. The online times for the client were obtained from a commodity laptop-like machine, c5.2xlarge.}
 \begin{tabular}{| c | c | c | c | c | c | c | c | c | c | c | c |} 
 \hline
  System & \# of steps & Security & Ciphertext & \multicolumn{2}{c|}{Max. runtime } & \multicolumn{2}{c|}{Max. runtime} & \multicolumn{2}{c|}{Max. runtime} & Max. & Server \\ [0.5ex]
  dim. & until refresh & level [bits] & modulus & \multicolumn{2}{c|}{phase I [s]} & \multicolumn{2}{c|}{phase II [s]} & \multicolumn{2}{c|}{phase III [s]} & RAM [GB] & machine \\ [0.5ex]
  \cline{5-10}
  & & & & Client & Server & Client & Server  & Client & Server & & \\  [0.5ex] \hline
  \hline
 $2\times2$ & 8 & 100 & 1120 & 0.88 & 3.48 & 0.84 & 220.2 & 0.23 & 4.36 & 17.19 & c5.4xlarge\\
\hline
 $2\times2$ & 5 & 142 & 802 & 0.6 & 3.23 & 0.5 & 298.1 & 0.17 & 6.21 & 12.57 & c5.2xlarge \\ \hline
 $4\times4$ & 8 & 100 & 1120 & 0.88 & 4.06 & 0.84 & 413.07& 0.23 & 4.94 & 34.32 & c5.9xlarge \\ \hline
 $4\times4$ & 5 & 142 & 802 & 0.6 & 2.95 & 0.5 & 325.72 & 0.17 & 4.42 & 27.29 & c5.9xlarge \\ \hline
\end{tabular}
\label{tab:simulations}
\end{table*}

\section{Conclusions and future work}
\label{sec:discussion}
We proposed an online data-driven control algorithm designed to be encryption-friendly, i.e., allow both good setpoint tracking and efficient encrypted computations for privacy-preserving control. We achieve this by choosing a regularized relaxed convex formulation of the control optimization problem, such that we obtain a closed form solution for the control input, which we prove to be close to the solution of the noiseless original problem, under conditions on the regularization parameters. Moreover, this solution can be computed using a low-depth arithmetic circuit using the Schur's complement formula for matrix inversion and manipulating the sequential matrix-vector operations into different forms. These computations are amenable for an efficient encrypted solution using leveled homomorphic encryption and carefully encoding vectors into ciphertexts. The client encrypts its measurements with such an encryption scheme and sends them to a cloud server, which returns the encrypted control input. We discussed how to achieve an optimized implementation of the encrypted computations and how to make sure that the computations can run for a desired number of time steps. Finally, we discussed extensions that can be implemented to enhance the functionality of this encrypted control algorithm.
These tools can be used to tackle other related problems, such as adaptive control or that involve sequential matrix updates.

The LQR data-driven approximation algorithm we presented in this paper is the first step towards private data-driven control. We plan to include constraints on the inputs and the outputs. This will substantially increase the complexity of the encrypted computations, since comparisons are not inherently polynomial operations. 
We plan to explore other encryption-friendly methods for private data-driven control, that are less data-hungry than the behavioral framework and more amenable for larger systems.

\section*{Acknowledgements}
This work was supported by the NSF under grant CPS 1837253. A. Alexandru would like to thank Yuriy Polyakov, Manfred Morari and Matei Ionita for helpful discussions.

\bibliographystyle{IEEEtran}
\bibliography{IEEEabrv,biblo}

\appendices
\section{Proof of closeness}
\label{app:theory}
In this section, we present a detailed proof of Theorem~\ref{thm:closeness}. First, we characterize the solution set of the original behavioral problem~\eqref{eq:optimization0}. Then, we characterize the solution of the approximate problem~\eqref{eq:optimization2}. To prove closeness, we make suitable use of the inversion lemma, the pseudo-inverse (denoted by $\dagger$) limit definition, and Singular Value Decomposition (SVD).
\subsection{Notation and preliminaries}
Before we proceed with the proof, we consider the following definitions:
\begin{align}
&\mbf{\hat{M}}:={\mbf Y^f}^\intercal\mbf{Q}\mbf{Y}^{f}+{\mbf U^f}^\intercal\mbf{R}\mbf{U}^{f}+\lambda_g \mbf{I}\label{eq:app_hat_M} \\ 
&\mbf{\hat{M}}_0:={\mbf Y^f}^\intercal\mbf{Q}\mbf{Y}^{f}+{\mbf U^f}^\intercal\mbf{R}\mbf{U}^{f}, \label{eq:app_hat_M_0}
\end{align}
where $\mbf{\hat{M}}$ is matrix $\mbf{M}$ defined in~\eqref{eq:M}, without the penalty terms. Similarly, $\mbf{\hat{M}}_0$ is matrix $\mbf{M}$ without any regularization or penalty terms. 

For compactness, we denote the matrix of past precollected data as $\mbf D_p$ and the vector of the last $M$ outputs and inputs as $\mbf d_t$:
\begin{equation}\label{eq:app_D_p}
    \mbf{D}_{p}:=\left[\begin{matrix}  \mbf U^p \\ \mbf Y^p \end{matrix}\right], \qquad
    \mbf{d}_t:=\left[\begin{matrix}  \mbf{\bar{u}}_t \\ \mbf{\bar{y}}_t \end{matrix}\right].
\end{equation}
Both matrices $\mbf{\hat{M}}_0,$ $\mbf{D}_p$ are singular. 
Let the SVD of the past precollected data $\mbf{D}_p$ be:
\begin{equation}\label{eq:app_svd_past} 
  \mbf{D}_p=\left[\begin{matrix}  \mbf{E}&\mbf{E}_{\perp} \end{matrix}\right] \left[\begin{matrix}  \mbf{\Sigma}&\mbf{0}\\\mbf{0}&\mbf{0} \end{matrix}\right] \left[\begin{matrix}  \mbf{F}^{\intercal}\\{\mbf{F}_{\perp}}^{\intercal} \end{matrix}\right] ,
\end{equation}
where $\mbf\Sigma\in\mathbb{R}^{q\times q}$ is an invertible diagonal matrix, containing the non-zero singular values of $\mbf{D}_p$ with $q=\mathrm{rank}(\mbf D_p)$, while matrices $\left[\begin{matrix}  \mbf{E}&\mbf{E}_{\perp} \end{matrix}\right]$, $\left[\begin{matrix}  \mbf{F}&\mbf{F}_{\perp} \end{matrix}\right]$ are orthonormal.
Recall that the set of optimal solutions of the original behavioral problem~\eqref{eq:optimization0} is denoted by $\mathcal{G}_{\mathrm{opt}}$--see~\eqref{eq:set_optimal_solutions}.
The minimum norm element of $ \mathcal{G}_{\mathrm{opt}}$ is denoted by $\mbf{g}_{\min}$--see~\eqref{eq:minimum_norm_g}. 

For completeness we present some inversion formulae.
\begin{lemma}[Inversion formulae]\label{lem:inversions}
Let $\mbf K,\mbf L,\mbf V$ be matrices of appropriate dimensions. Then:\\
a) Provided the respective inverses exist:
\begin{align*}
       & (\mbf K+\mbf{V}^{\intercal}\mbf L\mbf V)^{-1}=\mbf K^{-1}-\mbf K^{-1}\mbf V^{\intercal}(\mbf L^{-1}+\mbf V\mbf K^{-1}\mbf V^{\intercal})^{-1}\mbf V\mbf K^{-1}\\
        &(\mbf K+\mbf{V}^{\intercal}\mbf L\mbf V)^{-1}\mbf V^{\intercal}\mbf L=\mbf K^{-1}\mbf V^{\intercal}(\mbf L^{-1}+\mbf V\mbf K^{-1}\mbf V^{\intercal})^{-1}.
    \end{align*}
\noindent b) Let $\left[\begin{matrix}  \mbf{V}&\mbf{V}_{\perp} \end{matrix}\right]$ be an orthonormal matrix. Then, if $\mbf L$ is invertible:
    \begin{align}\label{eq:app_inversion_subspace}
     & (\mbf{V}^{\intercal}\mbf{L}^{-1}\mbf{V})^{-1}=\mbf{V}^{\intercal}\mbf{L}\mbf{V}-\mbf{V}^{\intercal}\mbf{L}\mbf{V}_{\perp}(\mbf{V}_{\perp}^{\intercal}\mbf{L}\mbf{V}_{\perp})^{-1}\mbf{V}_{\perp}^{\intercal}\mbf{L}\mbf{V} \\\label{eq:app_inversion_equivalence}
 &\mbf{I}-\mbf{V}_{\perp}(\mbf{V}^{\intercal}_{\perp}\mbf{L}\mbf{V}_{\perp})^{-1}\mbf{V}_{\perp}^\intercal\mbf L=\mbf L^{-1}\mbf{V}(\mbf{V}^{\intercal}\mbf{L}^{-1}\mbf{V})^{-1}\mbf{V}^{\intercal} .
    \end{align}
\end{lemma}
\begin{proof}
Part a) is standard, also known as the Woodbury matrix identity. To prove~\eqref{eq:app_inversion_subspace}, we use the identity $\mbf{V}_{\perp}\mbf{V}^{\intercal}_{\perp}+\mbf{V}\mbf{V}^{\intercal}=\mbf{I}$ and verify that the properties of the inverse hold.
To prove~\eqref{eq:app_inversion_equivalence}, we use~\eqref{eq:app_inversion_subspace}. 
\end{proof}

\subsection{Original behavioral problem}
In the following lemma, we characterize the solution set $\mathcal{G}_{\mathrm{opt}}$ of the original constrained behavioral problem~\eqref{eq:optimization0}. 
In general, there are infinite optimal solutions since $\mbf{\hat{M}}_0$ and $\mbf{D}_p$ have low rank.
\begin{lemma}\label{lem:solution_set_original}
    Consider optimization problem~\eqref{eq:optimization0}. Recall the SVD decomposition of $\mbf{D}_p$ in~\eqref{eq:app_svd_past}. Then
    \begin{itemize}
        \item[a)] $\mbf{g}\in \mathcal{G}_{\mathrm{opt}}$ is an optimal solution of~\eqref{eq:optimization0} if and only if $\mbf g$ lies in the following affine subspace:
\begin{equation}\label{eq:app_solution_linear_system}
\left[\begin{matrix} \mbf{\hat{M}}_0& \mbf{F}\\ \mbf{F}^{\intercal}&\mbf{0} \end{matrix}\right]\left[\begin{matrix} \mbf g\\ \bm{\mu} \end{matrix}\right]=\left[\begin{matrix} {\mbf Y^f}^\intercal\mbf{Q}\mbf{r}_t\\\mbf{\Sigma}^{-1}\mbf{E}^{\intercal}\mbf{d}_t \end{matrix}\right]
\end{equation}
for some $\bm{\mu}\in\mathbb{R}^{q},$
where $q$ is the rank of the past precollected data $\mbf{D}_p$.
\item[b)] Vector ${\mbf Y^f}^\intercal\mbf{Q}\mbf{r}_t$ lies in the column space of $\mbf{\hat{M}}_0$.
\item[c)] The above affine subspace can be equivalently described by the closed form expression:
\begin{align}
    \mbf g&=\mbf{g}_1+\mbf{g}_2+\mbf{g}_3,
\end{align}
where $\mbf{g}_i$, $i=1,2,3$ are orthogonal with each other and
\begin{align}\label{eq:app_g_1}
   \mbf{g}_1&=\mbf{F}\mbf{\Sigma}^{-1}\mbf{E}^{\intercal}\mbf{d}_t\\\label{eq:app_g_2}
   \mbf{g}_2&=\mbf{F}_{\perp}({\mbf{F}_{\perp}}^{\intercal}\mbf{\hat{M}}_0\mbf{F}_{\perp})^{\dagger}{\mbf{F}_{\perp}}^{\intercal}\left( {\mbf Y^f}^\intercal\mbf{Q}\mbf{r}_t-\mbf{\hat{M}}_0\mbf{g}_1\right)\\\label{eq:app_g_3}
   \mbf{g}_3&=\mbf{F}_{\perp}(\mbf I-({\mbf{F}_{\perp}}^{\intercal}\mbf{\hat{M}}_0\mbf{F}_{\perp})^{\dagger}{\mbf{F}_{\perp}}^{\intercal}\mbf{\hat{M}}_0\mbf{F}_{\perp})\mbf{s},
\end{align}
where $\mbf{s}$ is arbitrary.
\item[d)] As a result, the minimum norm element of $\mathcal{G}_{\mathrm{opt}}$ is
\begin{align}\label{eq:app_minimum_norm_g_alternative}
    \mbf{g}_{\min}=\mbf{g}_1+\mbf{g}_2.
\end{align}
    \end{itemize}
\end{lemma}
\begin{proof}
\textbf{Proof of a)}
Since $\mbf{d}_t$ lies in the range space of $\mbf{D}_p$, we can replace the singular equality constraint $\mbf{D}_p\mbf{g}=\mbf{d}_t$ by $\mbf{F}^{\intercal}\mbf{g}=\mbf{\Sigma}^{-1}\mbf{E}^{\intercal}\mbf{d}_t$, where now $\mbf{F}$ has full rank. If we also replace $\mbf{u}$, $\mbf{y}$ with $\mbf{U}^{f}\mbf{g}$ and $\mbf{Y}^{f}$ respectively, and we ignore the constant terms, then we obtain the optimization problem:
\begin{align}\label{eq:app_optimization_nonsingular}
\begin{split}
\min_{\mbf g}~&~\tfrac{1}{2}\mbf{g}^{\intercal}\mbf{\hat{M}}_0\mbf{g}-\mbf{g}^{\intercal}{\mbf Y^f}^\intercal\mbf{Q}\mbf{r}_t \\
s.t.~&~\mbf{F}^{\intercal}\mbf{g}=\mbf{\Sigma}^{-1}\mbf{E}^{\intercal}\mbf{d}_t,
\end{split}
\end{align}
which is equivalent to~\eqref{eq:optimization0}.
The Lagrangian of~\eqref{eq:app_optimization_nonsingular} is
\[
\mathcal{L}=\tfrac{1}{2}\mbf{g}^{\intercal}\mbf{\hat{M}}_0\mbf{g}-\mbf{g}^{\intercal}{\mbf Y^f}^\intercal\mbf{Q}\mbf{r}_t+\bm{\mu}^{\intercal}(\mbf{F}^{\intercal}\mbf{g}-\mbf{\Sigma}^{-1}\mbf{E}^{\intercal}\mbf{d}_t),
\]
where $\bm \mu\in\mathbb{R}^{q}$ are the Lagrange multipliers. The result follows by applying standard first order optimality conditions.

\textbf{Proof of b)} Observe that 
\[\mbf{\hat{M}}_0=\left[\begin{matrix} {\mbf Y^f}^\intercal\mbf{Q}^{1/2}& {\mbf U^f}^\intercal\mbf{R}^{1/2} \end{matrix}\right]\left[\begin{matrix} {\mbf Y^f}^\intercal\mbf{Q}^{1/2}& {\mbf U^f}^\intercal\mbf{R}^{1/2} \end{matrix}\right]^{\intercal}.\]
It is clear that ${\mbf Y^f}^\intercal\mbf{Q}\mbf{r}_t$ lies in the column space of $\left[\begin{matrix} {\mbf Y^f}^\intercal\mbf{Q}^{1/2}& {\mbf U^f}^\intercal\mbf{R}^{1/2} \end{matrix}\right]$. Then, using the fact that for a matrix $\mbf A$, $\mr{Range}(\mbf A) = \mr{Range}(\mbf A \mbf A^\intercal)$ (the proof easily follows from using SVD on $\mbf A$), we also show that ${\mbf Y^f}^\intercal\mbf{Q}\mbf{r}_t$ lies in the column space of $\hat{\mbf M}_0$.  

\textbf{Proof of c)}
First, we solve subequation:
\[\mbf{F}^{\intercal}\mbf{g}=\mbf{\Sigma}^{-1}\mbf{E}^{\intercal}\mbf{d}_t.\]
Since matrix $\left[\begin{matrix} \mbf{F}& \mbf{F}_{\perp}\end{matrix}\right]$ is an orthonormal basis, we can express all solutions $\mbf g$ as:
\begin{align}\label{eq:app_first_subequation}
\mbf{g}=\mbf{F}\mbf{\Sigma}^{-1}\mbf{E}^{\intercal}\mbf{d}_t+\mbf{F}_{\perp}\bm{\xi}=\mbf{g}_1+\mbf{F}_{\perp}\bm{\xi},
\end{align}
where $\mbf{F}_{\perp}\bm{\xi}$ captures the components of $\mbf{g}$ in the kernel of $\mbf{F}^{\intercal}$, and $\bm{\xi}$ is to be determined.

Second, we solve subequation:
\[
\mbf{\hat{M}}_0\mbf{g}+\mbf{F}\bm{\mu}={\mbf Y^f}^\intercal\mbf{Q}\mbf{r}_t.
\]
Multiplying from the left by the orthonormal matrix $\left[\begin{matrix}  \mbf{F}&\mbf{F}_{\perp} \end{matrix}\right]^{\intercal}$ results in the equivalent system of subequations:
\begin{align}
\mbf{F}^{\intercal}\mbf{\hat{M}}_0\mbf{g}+\bm{\mu}&=\mbf{F}^{\intercal}{\mbf Y^f}^\intercal\mbf{Q}\mbf{r}_t\label{eq:app_second_subequation_a}\\
{\mbf{F}_{\perp}}^{\intercal}\mbf{\hat{M}}_0\mbf{g}&={\mbf{F}_{\perp}}^{\intercal}{\mbf Y^f}^\intercal\mbf{Q}\mbf{r}_t\label{eq:app_second_subequation_b}.
\end{align}
The former subequation~\eqref{eq:app_second_subequation_a} just determines the Lagrange multiplier $\bm{\mu}$.  Replacing~\eqref{eq:app_first_subequation} into~\eqref{eq:app_second_subequation_b} gives:
\[
{\mbf{F}_{\perp}}^{\intercal}\mbf{\hat{M}}_0\mbf{F}_{\perp}\bm{\xi}={\mbf{F}_{\perp}}^{\intercal}{\mbf Y^f}^\intercal\mbf{Q}\mbf{r}_t-{\mbf{F}_{\perp}}^{\intercal}\mbf{\hat{M}}_0\mbf{g}_1.
\]
By the standard property of the pseudo-inverse~\cite[Ch.~2]{ben2003generalized}, we can write all possible solutions $\bm{\xi}$ as:
\begin{align}\label{eq:app_all_xi}
    \bm{\xi}=&\underbrace{({\mbf{F}_{\perp}}^{\intercal}\mbf{\hat{M}}_0\mbf{F}_{\perp})^{\dagger}{\mbf{F}_{\perp}}^{\intercal}\left( {\mbf Y^f}^\intercal\mbf{Q}\mbf{r}_t-\mbf{\hat{M}}_0\mbf{g}_1\right)}_{\bm{\xi}_1}+\underbrace{\left(\mbf{I}- ({\mbf{F}_{\perp}}^{\intercal}\mbf{\hat{M}}_0\mbf{F}_{\perp})^{\dagger}{\mbf{F}_{\perp}}^{\intercal}\mbf{\hat{M}}_0\mbf{F}_{\perp}\right)\mbf{s}}_{\bm{\xi}_2},
\end{align}
where $\mbf{s}$ is arbitrary. Moreover, the two components $\bm{\xi}_1$, $\bm{\xi}_2$ are orthogonal.
Combining~\eqref{eq:app_first_subequation} and~\eqref{eq:app_all_xi} gives:
\[
\mbf{g}=\mbf{g}_1+\mbf{F}_{\perp}\bm{\xi}_1+\mbf{F}_{\perp}\bm{\xi}_2,
\]
where all components are orthogonal to each other.

\textbf{Proof of d)}
Follows from b), orthogonality of the three summands and setting the only free parameter $\mbf s$ to zero. 
\end{proof}

\subsection{Approximate problem}
Due to the regularization term, the solution of the approximate problem~\eqref{eq:optimization2} is unique--see~\eqref{eq:opt_g}. We can show that if we chose large enough $\lambda$, then we can show that the solution of~\eqref{eq:optimization2} converges to an intermediate one.
\begin{lemma}\label{lem:solution_approximate}
    Consider optimization problem~\eqref{eq:optimization2} with $\lambda_u=\lambda_y=\lambda>0$ and $\lambda_g>0$ and let $\mbf{g}^\ast$ be the optimal solution. Define the intermediate solution:
\begin{align}\label{eq:app_g_bar}
&\mbf{\bar{g}}:=\mbf{\hat{M}}^{-1}\mbf{F}\left( \mbf{F}^{\intercal}\mbf{\hat{M}}^{-1}\mbf{F}\right)^{-1}\mbf{\Sigma}^{-1}\mbf{E}^{\intercal}\mbf{d}_t+\left[\mbf{\hat{M}}^{-1}-\mbf{\hat{M}}^{-1}\mbf{F}\left(\mbf{F}^{\intercal}\mbf{\hat{M}}^{-1}\mbf{F} \right)^{-1}\mbf{F}^{\intercal}\mbf{\hat{M}}^{-1}\right]{\mbf Y^f}^\intercal\mbf{Q}\mbf{r}_t.    
\end{align}
Then,
\begin{equation}\label{eq:app_g_bar_closeness}
    \left\lVert\mbf{g}^\ast-\mbf{\bar{g}}\right\rVert_2\le O(\lambda^{-1}),
\end{equation}
as $\lambda\rightarrow\infty$ and $\lambda_g\rightarrow 0^+$, where we used big-$O$ notation to hide quantities which do not depend on $\lambda,\lambda_g$.
\end{lemma}

\begin{proof}
\textbf{Step a). }
Recall that $\mbf{g}^\ast$ is given by~\eqref{eq:opt_g}. We will use the inversion lemma to bring $\mbf{g}^\ast$ in a form closer to $\mbf{\bar{g}}$.
Using the notation of this section,~\eqref{eq:opt_g} can be written as:
\begin{align*}
\mbf{g}^\ast&=\mbf{M}^{-1}\left({\mbf Y^f}^\intercal\mbf{Q}\mbf{r}_t+\lambda\mbf{D}_p \right)=(\mbf{\hat{M}}+\lambda \mbf{F}\mbf{\Sigma}^{2}\mbf{F}^{\intercal})^{-1}\left({\mbf Y^f}^\intercal\mbf{Q}\mbf{r}_t+\lambda\mbf{F}\mbf{\Sigma}\mbf{E}^{\intercal}\mbf{d}_t\right)\\
&=(\mbf{\hat{M}}+\lambda \mbf{F}\mbf{\Sigma}^{2}\mbf{F}^{\intercal})^{-1}\left({\mbf Y^f}^\intercal\mbf{Q}\mbf{r}_t+\lambda\mbf{F}\mbf{\Sigma}^2\mbf{F}^{\intercal}\mbf{g}_1\right),
\end{align*}
where the second equality follows from $\mbf{D}^{\intercal}_p\mbf{D}_p=\mbf{F}\mbf{\Sigma}^2\mbf{F}^{\intercal}$ and the last equality follows from $\mbf{F}^{\intercal}\mbf{F}=\mbf{I}$. By Lemma~\ref{lem:inversions} a):
\begin{align*}
(\mbf{\hat{M}}+\lambda \mbf{F}\mbf{\Sigma}^{2}\mbf{F}^{\intercal})^{-1}&=\mbf{\hat{M}}^{-1}-\mbf{\hat{M}}^{-1}\mbf{F}\left(\lambda^{-1}\mbf{\Sigma}^{-2}+\mbf{F}^{\intercal}\mbf{\hat{M}}^{-1}\mbf{F} \right)^{-1}\mbf{F}^{\intercal}\mbf{\hat{M}}^{-1},\\
(\mbf{\hat{M}}+\lambda \mbf{F}\mbf{\Sigma}^{2}\mbf{F}^{\intercal})^{-1}\mbf{F}\lambda\mbf{\Sigma}^2&=\mbf{\hat{M}}^{-1}\mbf{F}\left(\lambda^{-1}\mbf{\Sigma}^{-2}+\mbf{F}^{\intercal}\mbf{\hat{M}}^{-1}\mbf{F} \right)^{-1}.
\end{align*}

\textbf{Step b). } Next, define the difference:
\begin{align*}
\bm\Delta&:=\left(\lambda^{-1}\mbf{\Sigma}^2+\mbf{F}^{\intercal}\mbf{\hat{M}}^{-1}\mbf{F}\right)^{-1}-\left(\mbf{F}^{\intercal}\mbf{\hat{M}}^{-1}\mbf{F}\right)^{-1} =\lambda^{-1}\left(\mbf{F}^{\intercal}\mbf{\hat{M}}^{-1}\mbf{F}\right)^{-1}\mbf{\Sigma}^{2}\left(\lambda^{-1}\mbf{\Sigma}^2+\mbf{F}^{\intercal}\mbf{\hat{M}}^{-1}\mbf{F}\right)^{-1},
\end{align*}
where the equality follows from $\mbf A^{-1}+\mbf B^{-1} = \mbf B^{-1}(\mbf A+\mbf B)\mbf A^{-1}$. The error  $\mbf{g}^\ast-\mbf{\bar{g}}$ can be written as:
\begin{align*}
    \mbf{g}^\ast-\mbf{\bar{g}}=&-\mbf{\hat{M}}^{-1}\mbf{F}\bm{\Delta}\mbf{F}^{\intercal}\mbf{\hat{M}}^{-1}{\mbf Y^f}^\intercal\mbf{Q}\mbf{r}_t+\mbf{\hat{M}}^{-1}\mbf{F}\bm{\Delta}\mbf{\Sigma}^{-1}\mbf{E}^{\intercal}\mbf{d}_t.
\end{align*}
To complete the proof, it is sufficient to show that the matrices $\mbf{\hat{M}}^{-1}\mbf{F}\left(\mbf{F}^{\intercal}\mbf{\hat{M}}^{-1}\mbf{F}\right)^{-1}$, $\left(\lambda^{-1}\mbf{\Sigma}^2+\mbf{F}^{\intercal}\mbf{\hat{M}}^{-1}\mbf{F}\right)^{-1}$, and $\mbf{\hat{M}}^{-1}{\mbf Y^f}^\intercal\mbf{Q}\mbf{r}_t$ are bounded as $\lambda_g\rightarrow 0^+$, $\lambda\rightarrow \infty$.

\textbf{Step c)} Boundedness of $\mbf{\hat{M}}^{-1}\mbf{F}\left(\mbf{F}^{\intercal}\mbf{\hat{M}}^{-1}\mbf{F}\right)^{-1}$ follows by Lemma~\ref{lem:solution_intermediate}.

\textbf{Step d)}
The following matrix is always bounded: 
\[
\left\lVert\left(\lambda^{-1}\mbf{\Sigma}^2+\mbf{F}^{\intercal}\mbf{\hat{M}}^{-1}\mbf{F}\right)^{-1}\right\rVert_2\le O(1).
\]
It follows from $\mbf{\hat{M}}\preceq (\lambda_g+\lVert \mbf{\hat{M}}_0 \rVert_{2})\mbf{I}$, which in turn implies
\[\lambda^{-1}\mbf{\Sigma}^2+\mbf{F}^{\intercal}\mbf{\hat{M}}^{-1}\mbf{F}\succeq \mbf{F}^{\intercal}\mbf{\hat{M}}^{-1}\mbf{F}\succeq  (\lambda_g+\lVert \mbf{\hat{M}}_0 \rVert_{2})^{-1}\mbf{I},\]
where $\preceq,\succeq$ denote comparison with respect to the positive semi-definite cone.

\textbf{Step e). } The norm of $\mbf{\hat{M}}^{-1}{\mbf Y^f}^\intercal\mbf{Q}\mbf{r}_t$ is $O(1)$. From Lemma~\ref{lem:solution_set_original} b), ${\mbf Y^f}^\intercal\mbf{Q}\mbf{r}_t$ lies in the range space of $\mbf{\hat{M}}_0$, which allows us to write
$
{\mbf Y^f}^\intercal\mbf{Q}\mbf{r}_t
=\mbf{\hat{M}}_0\mbf{\hat{M}}^{\dagger}_0{\mbf Y^f}^\intercal\mbf{Q}\mbf{r}_t$.
Hence, 
\[
\lVert\mbf{\hat{M}}^{-1}{\mbf Y^f}^\intercal\mbf{Q}\mbf{r}_t\rVert_2\le O(\lVert\mbf{\hat{M}}^{-1}\mbf{\hat{M}}_0\rVert_2)=O(1),
\]
since $\mbf{\hat{M}}=\mbf{\hat{M}}_0+\lambda_g\mbf{I}\succeq \mbf{\hat{M}}_0$, and $\lVert\mbf{\hat{M}}^{-1}\mbf{\hat{M}}_0\rVert_2\le 1$.
\end{proof}

\subsection{Proof of Theorem~\ref{thm:closeness}}
Consider the intermediate solution $\mbf{\bar{g}}$ defined in~\eqref{eq:app_g_bar}. By the triangle inequality we have:
\[
\lVert \mbf{g}_{\min}-\mbf{g}^*\rVert_2\le\lVert \mbf{g}_{\min}-\mbf{\bar{g}}\rVert_2+\lVert \mbf{\bar{g}}-\mbf{g}^*\rVert_2.
\]
In Lemma~\ref{lem:solution_approximate}, we show that the second error term decays to zero as fast as $\lambda^{-1}$. To complete the proof, we invoke the following Lemma, which states that the first error term decays to zero as well.\hfill $\qedsymbol$
\begin{lemma}\label{lem:solution_intermediate} 
Consider the minimum norm solution $\mbf{g}_{\min}$ defined in~\eqref{eq:minimum_norm_g} and the intermediate solution $\mbf{\bar{g}}$ defined in~\eqref{eq:app_g_bar}.  Let $\lambda_g>0$. Then:\\
a) The following limit holds
\begin{equation}\label{eq:app_pseudo_inverse_limit}
    \lim_{\lambda_g\rightarrow 0^+}({\mbf{F}_{\perp}}^{\intercal}\mbf{\hat{M}}\mbf{F}_{\perp})^{-1}{\mbf{F}_{\perp}}^{\intercal}\mbf{\hat{M}}_0=({\mbf{F}_{\perp}}^{\intercal}\mbf{\hat{M}}_0\mbf{F}_{\perp})^{\dagger}{\mbf{F}_{\perp}}^{\intercal}\mbf{\hat{M}}_0.
\end{equation}
b) The following matrix is bounded
\begin{equation}\label{eq:app_boundedness}
   \left\lVert\mbf{\hat{M}}^{-1}\mbf{F}\left(\mbf{F}^{\intercal}\mbf{\hat{M}}^{-1}\mbf{F}\right)^{-1}\right\rVert_2=O(1).
\end{equation}
c) The intermediate solution converges to the minimum norm solution as $\lambda_g$ goes to zero
\begin{equation}\label{eq:app_intermediate_convergence}
    \lim_{\lambda_g\rightarrow 0^+}\lVert \mbf{g}_{\min}-\mbf{\bar{g}}\rVert_2=0.
\end{equation}
\end{lemma}

\begin{proof}
\textbf{Proof of a). } Let $\mbf T:=\mbf{\hat{M}}_0^{1/2}\mbf{F}_{\perp}$. Since ${\mbf{F}_{\perp}}^{\intercal}\mbf{F}_{\perp}=\mbf{I}$:
\[
({\mbf{F}_{\perp}}^{\intercal}\mbf{\hat{M}}\mbf{F}_{\perp})^{-1}{\mbf{F}_{\perp}}^{\intercal}\mbf{\hat{M}}_0=(\lambda_g\mbf{I}+\mbf{T}^{\intercal}\mbf{T})^{-1}\mbf{T}^{\intercal}\mbf{\hat{M}}_0^{1/2}.
\]
The result follows from the pseudo-inverse limit definition and $\mbf T^\dagger = (\mbf T^{\intercal}\mbf{T})^{\dagger}\mbf T^{\intercal}$~\cite[Ch~1]{ben2003generalized}.

\textbf{Proof of b). } Note that the norm remains unchanged if we multiply from the right by $\mbf{F}^{\intercal}$. 
\begin{align*}
& \left\lVert\mbf{\hat{M}}^{-1}\mbf{F}\left(\mbf{F}^{\intercal}\mbf{\hat{M}}^{-1}\mbf{F}\right)^{-1}\right\rVert_2= \left\lVert\mbf{\hat{M}}^{-1}\mbf{F}\left(\mbf{F}^{\intercal}\mbf{\hat{M}}^{-1}\mbf{F}\right)^{-1}\mbf{F}^{\intercal}\right\rVert_2
\end{align*}
By~\eqref{eq:app_inversion_equivalence} and the triangle inequality
\begin{align*}
&\left\lVert\mbf{\hat{M}}^{-1}\mbf{F}\left(\mbf{F}^{\intercal}\mbf{\hat{M}}^{-1}\mbf{F}\right)^{-1}\right\rVert_2=\left\lVert\mbf{I}-\mbf{F}_{\perp}(\mbf{F}_{\perp}^{\intercal}\mbf{\hat{M}}\mbf{F}_{\perp})^{-1}\mbf{F}_{\perp}^{\intercal}\mbf{\hat{M}}\right\rVert_2\\
&\le 1+\left\lVert\mbf{F}_{\perp}({\mbf{F}_{\perp}}^{\intercal}\mbf{\hat{M}}\mbf{F}_{\perp})^{-1}{\mbf{F}_{\perp}}^{\intercal}\mbf{\hat{M}}\right\rVert_2=1+\left\lVert\mbf{F}_{\perp}({\mbf{F}_{\perp}}^{\intercal}\mbf{\hat{M}}\mbf{F}_{\perp})^{-1}{\mbf{F}_{\perp}}^{\intercal}\mbf{\hat{M}}\left[\begin{matrix} \mbf F& \mbf{F}_{\perp} \end{matrix}\right]\right\rVert_2,
\end{align*}
where the last equality follows from the fact that multiplying from the right by $\left[\begin{matrix} \mbf F& \mbf{F}_{\perp} \end{matrix}\right]$ leaves the norm unchanged.
Submatrix $\mbf{F}_{\perp}({\mbf{F}_{\perp}}^{\intercal}\mbf{\hat{M}}\mbf{F}_{\perp})^{-1}{\mbf{F}_{\perp}}^{\intercal}\mbf{\hat{M}}\mbf{F}_{\perp}=\mbf{F}_{\perp}$ has bounded norm. For the other submatrix:
\begin{align*}
\mbf{F}_{\perp}({\mbf{F}_{\perp}}^{\intercal}\mbf{\hat{M}}\mbf{F}_{\perp})^{-1}{\mbf{F}_{\perp}}^{\intercal}\mbf{\hat{M}}\mbf{F}=\mbf{F}_{\perp}({\mbf{F}_{\perp}}^{\intercal}\mbf{\hat{M}}\mbf{F}_{\perp})^{-1}{\mbf{F}_{\perp}}^{\intercal}\mbf{\hat{M}}_0\mbf{F},
\end{align*}
since by orthogonality ${\mbf{F}_{\perp}}^{\intercal}\mbf{\hat{M}}\mbf{F}={\mbf{F}_{\perp}}^{\intercal}\mbf{\hat{M}}_0\mbf{F}$. By~\eqref{eq:app_pseudo_inverse_limit}, we immediately obtain boundedness since:
\begin{align}\label{eq:app_limit_component_of_g}
&\lim_{\lambda_g\rightarrow 0^+}({\mbf{F}_{\perp}}^{\intercal}\mbf{\hat{M}}\mbf{F}_{\perp})^{-1}{\mbf{F}_{\perp}}^{\intercal}\mbf{\hat{M}}\mbf{F}= ({\mbf{F}_{\perp}}^{\intercal}\mbf{\hat{M}}_0\mbf{F}_{\perp})^{\dagger}{\mbf{F}_{\perp}}^{\intercal}\mbf{\hat{M}}_0\mbf{F}.
\end{align}
Hence, both submatrices have bounded norm, which implies boundedness for the original matrix.

\textbf{Proof of c). }Using~\eqref{eq:app_inversion_equivalence}, we can rewrite $\mbf{\bar{g}}$ in a form that resembles $\mbf{g}_{\min}$:
\begin{align}\label{eq:app_g_bar_alternative}
&\mbf{\bar{g}}:=\mbf{F}\mbf{\Sigma}^{-1}\mbf{E}^{\intercal}\mbf{d}_t-{\mbf{F}_{\perp}}^{\intercal}({\mbf{F}_{\perp}}^{\intercal}\mbf{\hat{M}}\mbf{F}_{\perp})^{-1}{\mbf{F}_{\perp}}^{\intercal}\mbf{\hat{M}}\mbf{F}\mbf{\Sigma}^{-1}\mbf{E}^{\intercal}\mbf{d}_t+{\mbf{F}_{\perp}}^{\intercal}({\mbf{F}_{\perp}}^{\intercal}\mbf{\hat{M}}\mbf{F}_{\perp})^{-1}{\mbf{F}_{\perp}}^{\intercal}{\mbf Y^f}^\intercal\mbf{Q}\mbf{r}_t.
\end{align}
By~\eqref{eq:app_limit_component_of_g}, we get that the second term converges to $
-{\mbf{F}_{\perp}}^{\intercal}({\mbf{F}_{\perp}}^{\intercal}\mbf{\hat{M}}_0\mbf{F}_{\perp})^{\dagger}{\mbf{F}_{\perp}}^{\intercal}\mbf{\hat{M}}_0\mbf{F}\mbf{\Sigma}^{-1}\mbf{E}^{\intercal}\mbf{d}_t$.
For the third term, recall from Lemma~\ref{lem:solution_set_original} b), that ${\mbf Y^f}^\intercal\mbf{Q}\mbf{r}_t$ lies in the column space of $\mbf{\hat{M}}_0$, which implies that 
${\mbf Y^f}^\intercal\mbf{Q}\mbf{r}_t=\mbf{\hat{M}}_0\mbf{\hat{M}}^{\dagger}_0{\mbf Y^f}^\intercal\mbf{Q}\mbf{r}_t$. Using this identity and~\eqref{eq:app_pseudo_inverse_limit}, we also get that the third term converges to ${\mbf{F}_{\perp}}^{\intercal}({\mbf{F}_{\perp}}^{\intercal}\mbf{\hat{M}}_0\mbf{F}_{\perp})^{\dagger}{\mbf{F}_{\perp}}^{\intercal}{\mbf Y^f}^\intercal\mbf{Q}\mbf{r}_t$.

The result follows by~\eqref{eq:app_minimum_norm_g_alternative},~\eqref{eq:app_g_bar_alternative} and the convergence of the aforementioned terms.
\end{proof}

\allowdisplaybreaks
\section{CKKS scheme}
\label{app:CKKS}

We present a schematic view of the CKKS leveled homomorphic cryptosystem, introduced in~\cite{Cheon2017homomorphic}. 
The plaintext space is the ring of polynomials $\mcR := \mbZ[X]/\langle \Phi_M(X)\rangle$, where $\Phi_M(X)$ is the $M$-th cyclotomic polynomial of degree $N:=\phi(M)/2$, i.e., $X^N+1$. $N$ is also called the ring dimension for the plaintext space. The ciphertext space for level $l$ is the ring $(\mbZ_{Q_l}[X]/\langle \Phi_M(X)\rangle)^2$, where $Q_l$ is the ciphertext modulus corresponding to level $l$. Define also the residue ring $\mcR_l := \mcR/Q_l\mcR$.

The CKKS scheme supports an efficient packing of up to $N/2$ real values in a single plaintext, respectively a single ciphertext. A real-valued message in $\mbC^{N}$ is encoded in a plaintext in $\mbZ[X]/\langle \Phi_M(X)\rangle$ via the inverse canonical embedding map and rounding--which corresponds to evaluation of the inverse Discrete Fourier Transform (with the primitive $M$-th roots of unity $\mr{exp}(-2\pi i/M)$). The decoding procedure is given by the forward Discrete Fourier Transform. This batching technique allows performing Single Instruction Multiple Data (SIMD) operations on the ciphertexts.

Each operation introduces some noise in the ciphertext, see~\cite{Cheon2017homomorphic} for the magnitude of the different noises. The multiplication operation (along with rescaling) introduces the largest noise. Hence, the multiplicative budget is set from the beginning as a parameter of the scheme, and correct decryption is only guaranteed for the evaluation of circuits of smaller multiplicative depth. 
To deal with precision, in the CKKS scheme, a scaling factor is multiplied to each value. After a multiplication (including multiplication by plaintexts), one scaling factor needs to be removed in order to keep the correctness. This is done through a rescaling procedure which decreases the ciphertext modulus. 

The semantic security of the CKKS scheme assumes that the Decisional Ring Learning with Errors problem is computationally hard~\cite{Lyubashevsky2010ideal}.

We next sketch the algorithms in the basic CKKS scheme. The implementation we use~\cite{PALISADE} employs the Residue Number System optimization detailed in~\cite{Cheon2018full,Halevi2019improved,Blatt2020optimized}.

$\texttt{Setup}(1^\kappa,Q_L)$ Given the security parameter $\kappa$, for the largest ciphertext modulus $Q_L$, output the ring dimension $N$. Set the small distributions $\chi_{key}, \chi_{err}, \chi_{enc}$ over $\mcR$ for the secret, error and encryption.

$\texttt{KeyGen}(\mcR,\chi_{key}, \chi_{err}, \chi_{enc})$: Sample a secret $s\leftarrow \chi_{key}$, a random $a\leftarrow \mcR_L$ and error $e\leftarrow \chi_{err}$. Set the secret key $\mbf{sk} \leftarrow (1,s)$ and output the public key $\mbf{pk}\leftarrow (b,a)\in\mcR^2_L$, where $b\leftarrow -as+e(\bmod Q_L)$.

$\texttt{KSGen}_{\mbf{sk}}(\mcR,\chi_{key}, \chi_{err}, \chi_{enc})$: For $s'\in\mcR$, sample a random $a'\leftarrow \mcR_{2L}$ and error $e'\leftarrow \chi_{err}$. Output the switching key $\mbf{swk}\leftarrow (b',a')\in\mcR^2_{2L}$, where $b'\leftarrow -a's+e'+Q_Ls'(\bmod Q_{2L})$. Output the evaluation key $\mbf{evk}\leftarrow \texttt{KSGen}_{\mbf{sk}}(s^2)$. Output the rotation key $\mbf{rk}^{(\kappa)} \leftarrow \texttt{KSGen}_{\mbf{sk}}(\mbf s^{(\kappa)})$, for rotation index $\kappa$.

$\texttt{Enc}_{\mbf{pk}}(m)$: For $m\in\mcR$, sample $v\leftarrow \chi_{enc}$ and $e_0,e_1\leftarrow \chi_{err}$. Output $\mbf{ct}\leftarrow v\cdot\mbf{pk} + (m+e_0,e_1) (\bmod Q_L)$.

$\texttt{Dec}_{\mbf{sk}}(\mbf{ct})$: For $\mbf{ct} \leftarrow (c_0,c_1) \in\mcR_l^2$, output $\tilde m = c_0 + c_1 \cdot s (\bmod Q_l)$.

$\texttt{CAdd}(\mbf{ct},c)$: For $\mbf{ct} = (c_0,c_1) \in\mcR_l^2$ and $c\in\mcR$, output $\mbf{ct}_{CAdd} \leftarrow (c_0 + c, c_1) (\bmod Q_l)$.

$\texttt{Add}(\mbf{ct}_1,\mbf{ct}_2)$: For $\mbf{ct}_1, \mbf{ct}_2 \in \mcR^2_l$, output $\mbf{ct}_{Add} \leftarrow \mbf{ct}_1 + \mbf{ct}_2 (\bmod Q_l)$.

$\texttt{CMult}(\mbf{ct},c)$: For $\mbf{ct}\in\mcR_l^2$ and $c\in\mcR$, output $\mbf{ct}_{CMult} \leftarrow c\cdot \mbf{ct} (\bmod Q_l)$.

$\texttt{Mult}_{\mbf{evk}}(\mbf{ct}_1,\mbf{ct}_2)$: For $\mbf{ct}_i = (b_i,a_i)\in \mcR^2_l$, for $i=1,2$, let $(d_0,d_1,d_2) = (b_1b_2, a_1b_2 + a_2b_1, a_1a_2) (\bmod Q_l)$. Output $\mbf{ct}_{Mult} \leftarrow (d_0,d_1) + \lfloor Q_L^{-1}\cdot d_2 \cdot \mbf{evk} \rceil (\bmod Q_l)$.

$\texttt{Rotate}_{\mbf{rk}^{(\kappa)}}(\mbf{ct})$: For $\mbf{ct} = (c_0,c_1) \in \mcR^2_l$ and rotation~index $\kappa$, output $\mbf{ct}_{rot} \leftarrow (c_0^{(\kappa)}, 0) + \lfloor Q_L^{-1}\cdot c_1^{(\kappa)} \cdot \mbf{rk}^{(\kappa)} \rceil (\bmod Q_l)$.

$\texttt{ReScale}(\mbf{ct},p)$: For a ciphertext $\mbf{ct}\in\mcR^2_l$ and an integer $p$, output $\mbf{ct}'\leftarrow \lfloor 2^{-p}\cdot \mbf{ct}\rceil(\bmod Q_l/2^p)$.
\allowdisplaybreaks
\section{Implementation details for online solution}
\label{app:online}

\subsection{Ciphertext packing}
\label{appsub:packing}

Assume without loss of generality that we shift the time axis to the left by $M+N$, such that the trajectory concatenation mentioned in Section~\ref{subsec:online} is performed before $t=0$. This will simplify the circuit depth expression. 
In the following, we will address some of the encoding methods we investigated, with the goals to minimize the depth of the resulting circuit and to reduce the memory consumption and time complexity.

\subsubsection{Scalar encoding}
\label{subsubsec:unpacked}
each ciphertext only encodes and encrypts a scalar, i.e., element of a vector or matrix. This version has the smallest multiplicative depth compared to the other packing methods, no rotation key storage, but the largest ciphertext storage and the largest number of operations. It returns $m$ ciphertexts instead of one to the client, unless we add more processing to mask and rotate the elements of $\mbf u_t$ to obtain a single ciphertext. 

\subsubsection{``Hybrid" encoding}
\label{subsubsec:hybrid}
some ciphertexts encode and encrypt individual elements while other ciphertexts encode and encrypt vectors. Specifically, we pack the input and output vectors and columns of input and output matrices, but hold the entries of the matrix $\mbf M_t^{-1}$ and intermediate vectors $\mbf m_t$ as individual ciphertexts. This version returns only one vector for the client to decrypt and involves the circuit depth given in~\eqref{eq:depthRule}. 

The ``hybrid" encoding was implemented in~\cite{Alexandru2020towards}. The conclusions there were that for the best compromise of depth versus client computation and communication, we ask the client to send back an encryption of $1/s_t$, along with an encryption of the input $\mbf u_t$ and output $\mbf y_t$. This gives the depth expression of the denominator and numerator of $\mbf u_t$, for $t>1$):
\begin{align}\label{eq:depthRule}
\begin{split}
	d(D\mbf u_t) = d(s_t) &= 2t + 1, \quad
	d(N\mbf u_t) = 2t + 4.
\end{split}
\end{align}

However, this ``hybrid" solution still incurred a large storage: namely $O((S+t)^2)$ ciphertexts and $O((S+t)^2)$ rotation keys (only $O(S+t)$ rotation keys when there is no refreshing and packing of the matrix $\mbf M_t^{-1}$ required), and $O((S+t)^2)$ computations at the cloud, with large hidden constants. This suggests the need for a more tractable solution, involving a more efficient encoding.

\subsubsection{Vector encoding}
\label{subsubsec:packed}
each ciphertext encodes and encrypts multiple values, i.e. a vector or a column matrix. Via this type of encoding, we aim to keep the minimum possible depth as in the previous two versions, but minimize the memory requirements for ciphertexts, from $O((S+t)^2)$ to $O(S+t)$. At the same time, we want to decrease the rotation key storage from $O((S+t)^2)$ to $O(S+t)$ and keep the number of computations at number of computations at $O((S+t)^2)$.

With the previous ``hybrid" encoding on the inputs in Appendix~\ref{subsubsec:hybrid}, used in~\cite{Alexandru2020towards}, we found it preferable to encode the entries of $\mbf M_t^{-1}$ in individual ciphertexts in order to avoid masking when computing the desired result, which would have increased the depth. However, we found that using redundancy in the stored variables, i.e., repeating the elements in vector encoding of the input columns, along with completely rewriting the way the operations are performed, can alleviate this issue and keep the same depth as before.

We will use the encryption of the repeated elements encoding, $\mr{E_{vr0}}(\cdot)$, for the columns of the Hankel matrices $H\mbf U_t$ and $H\mbf Y_t$, the vectors $\bar{\mbf u}_t, \bar{\mbf y}_t$, the inputs and outputs $\mbf u_t, \mbf y_t$, the reference $\mbf r_t$, and to encode the costs. In the encoding, each element is repeated for $S+\bar T$ times, where $\bar T$ is the maximum number of online collected samples. 

For $t=-M-N+1:0$, follow the steps in the offline feedback solution, cf. Section~\ref{sec:naSolution}.

Let $S' := S + t-1$, for $t\geq 1$. The following plaintexts and ciphertexts are stored at the cloud: 
$\lambda_y$ and $\mbf Q$ encoded as $\mr{e_{vr0}}(\lambda\mbf Q):=$ $\mr{e_{vr0}}\left(\left[\begin{matrix}{\lambda_y} & \ldots & {\lambda_y} & q_{[0]} & \ldots & q_{[pN-1]} \end{matrix}\right]\right)$, $\lambda_u$ and $\mbf R$ encoded as $\mr{e_{vr0}}(\lambda\mbf R):=\mr{e_{vr0}}(\left[\begin{matrix}{\lambda_u} & \ldots & {\lambda_u} \end{matrix} \right. $ $\left. \begin{matrix} r_{[0]} & \ldots & r_{[mN-1]} \end{matrix}\right])$, $\lambda_g$ encoded as $\mr{e_{vr0}}(\lambda_g)$,
$\mr{E_{vr0}}(\mbf y_t)$, 
$\mr{E_{vr0}}(\mbf u_t)$, 
$\mr{E_{vr0}}(\mbf r_t)$,
$\mr{E_{vr0}}(\bar{\mbf y}_{t-1})$,
$\mr{E_{vr0}}(\bar{\mbf u}_{t-1})$ 
and for $i\in [S'+1]$:  
$\mr{E_{v0}}(\mr{col}_i(\mbf M_{t-1}^{-1}))$, 
$\mr{E_{vr0}}(\mr{col}_i(H\mbf Y_{t})) $,
$\mr{E_{vr0}}(\mr{col}_i(H\mbf U_{t})) $. 
To avoid some masking operations, the cloud also stores the rows of $\mbf U_t^f$: $\mr{E_{v0}}(\mr{row}_i(\mbf U_t^f))$ for $i\in[mN]$.

The cloud service then locally computes the following, corresponding to lines 9 and 13 of Algorithm~\ref{alg:online}. 
\begin{align*}
 	&\mr{E_{vr0}}(\bar{\mbf y}_t) = \rho(\mr{E_{vr0}}(\bar{\mbf y}_{t-1}), p(S'+1)) + \rho(\mr{E_{vr0}}(\mbf y_t), (1-p)M(S'+1)) \\
 	&\mr{E_{vr0}}(\bar{\mbf u}_t) = \rho(\mr{E_{vr0}}(\bar{\mbf u}_{t-1}), m(S'+1)) + \rho(\mr{E_{vr0}}(\mbf u_t), (1-m)M(S'+1))\\
	&\mr{E_{vr0}}(\mr{col}_{S'+1}H\mbf Y_t) = \rho(\mr{E_{vr0}}(\mr{col}_{S'}H\mbf Y_t), p(S'+1))+ \rho(\mr{E_{vr0}}(\mbf y_{t}), -p(M+N-1)(S'+1)) \\
	&\mr{E_{vr0}}(\mr{col}_{S'+1}H\mbf U_t) = \rho(\mr{E_{vr0}}(\mr{col}_{S'}H\mbf U_t), m(S'+1)) + \rho(\mr{E_{vr0}}(\mbf u_{t}), -m(M+N-1)(S'+1))\\
	&\mr{E_{v0}}(\mr{row}_i(\mbf U_t^f))= \mr{E_{v0}}(\mr{row}_i(\mbf U_t^f)) +\rho(\mr{E_{vr0}}(\mr{col}_{S'+1}(H\mbf U_t))\odot \mbf e_{i(S+\bar T)}, S'+1), i\in[m]\\
	&\mr{E_{vr0}}(\mbf {\bar y}_t, \mbf r_t) = \mr{E_{vr0}}(\mbf {\bar y}_t) + \rho(\mr{E_{vr0}}(\mbf r_t), -pM(S'+1)).
\end{align*}

The precollected input and output measurements have a multiplicative depth of 0. To reduce the total circuit depth, at time $t+1$, although the cloud service has computed the encryption of $\mbf u_t$, the client sends a fresh encryption of it. This allows us to have, $ \forall t \geq 0$:
\[d(\mr{cols}H\mbf U_t) = d(\mr{cols}H\mbf Y_t) = d(\bar{\mbf u}_t) = d(\bar{\mbf y}_t) = d(\mbf r_t) = 0.\]
The update of the row representation of $\mbf U^f$ needs a masking, so $d(\mr{rows}\mbf U_f) = 1$. 
From the offline computations, we assume that we have fresh encryptions for the columns of $\mbf M_0^{-1}$, hence $d(\mr{cols}\mbf M_0^{-1}) = 0$.

The encrypted operations, their flow and their motivations are described in Section~\ref{subsec:reducing}. 
Compared to the previous approach in~\cite{Alexandru2020towards}, we observed that we can avoid extra computations and a possible extra level in the computations if the server asks the client to compute $1/s_t$ immediately after it computes $s_t$. 
Hence, the server can compute directly the encryptions of $\mbf M_t^{-1}$ and $\mbf u_t$ instead of the numerators $N\mbf M_t^{-1}$ and $N\mbf u_t$. The encrypted operations are described in~\eqref{eq:encrypted_comp_vector}--\eqref{eq:encrypted_comp_vector3}, which correspond to lines 6 and 7 from Algorithm~\ref{alg:online}. 
$\mr{EvalSum}_{S+\bar T}$ refers to rotating and summing batches of $S+\bar T$ slots in the ciphertexts and $\mbf v_t := (\mbf m_t\mbf M_{t-1}^{-1})^\intercal$.
\begin{flalign*}
 	&\mr{E_{vr\ast}}(\mu_t) = \mr{EvalSum}_{S+\bar T} \big(\mr{E_{vr0}}(\mr{col}_{S'+1}H\mbf Y_{t-1})\odot \mr{e_{vr0}}(\lambda\mbf Q) \odot \mr{E_{vr0}}(\mr{col}_{S'+1}H\mbf Y_{t-1}) +\\
 	&\quad+\mr{E_{vr0}}(\mr{col}_{S'+1}H\mbf U_{t-1})\odot \mr{e_{vr0}}(\lambda\mbf R)\odot\mr{E_{vr0}}(\mr{col}_{S'+1}H\mbf U_{t-1}) + \mr{e_{vr0}}(\lambda_g) \\
 	&\mr{E_{vr\ast}}({\mbf m_t}_{[i]}) =  \mr{EvalSum}_{S+\bar T} \big(\mr{E_{vr0}}(\mr{col}_{S'+1}H\mbf Y_{t-1})\odot \mr{e_{vr0}}(\lambda\mbf Q)\odot \mr{E_{vr0}}(\mr{col}_iH\mbf Y_{t-1}) + \\
 	&\quad+\mr{E_{vr0}}(\mr{col}_{S'+1}H\mbf U_{t-1})\odot \mr{e_{vr0}}(\lambda\mbf R) \odot\mr{E_{vr0}}(\mr{col}_iH\mbf U_{t-1}),~i \in [S'+1] \\
	&\mr{E_{v0}(\mbf v_t)} = \sum_{i = 0}^{S'} \mr{E_{v0}}(\mr{col}_i(\mbf M_{t-1}^{-1})) \odot \mr{E_{vr\ast}}({\mbf m_t}_{[i]}) \numberthis\label{eq:encrypted_comp_vector}\\
 	&\mr{E_{vr\ast}}(s_t) =\mr{E_{vr\ast}}(\mu_t)- \sum_{i=0}^{S'} \mr{E_{vr\ast}}({\mbf m_t}_{[i]}) \odot \rho\left(\mr{E_{v0}}(\mbf v_t),i\right). \\[0.5ex]
	&\text{After receiving a fresh encryption of $1/s_t$ from the client:}\\[0.5ex]
 	&\mr{E_{v0}}(\mr{col}_{S'+1}(\mbf M_{t}^{-1})) = \left(-\mr{E_{v0}}(\mbf v_t) + \mbf{e}_{S'+1}\right)\odot \mr{E_{vr0}}(1/s_t) \\
	&\mr{E_{v0}}(\mbf w_i) = \sum_{j=0}^{S'} \mr{E_{v0}}(\mr{col}_j \mbf M_{t-1}^{-1}) \odot \left(\mr{E_{vr\ast}}({\mbf m_t}_{[i]})  \odot \mr{E_{v0}}(1/s_t\,\mbf e_i )\right)\\
 	&\mr{E_{v0}}(\mr{col}_i1/s_t\mbf v_t\mbf v_t^\intercal)= \mr{E_{v0}}(\mbf v_t) \odot \sum_{j=0}^{S'}  \rho\left(\mr{E_{v0}}(\mbf w_i), i-j \right),~i \in [S'+1] \\
	&\mr{E_{v0}}(\mr{col}_{i}(\mbf M_{t}^{-1})) = \mr{E_{v0}}(\mr{col}_{i}\mbf M_{t}^{-1}) - \mr{E_{v0}}(\mr{col}_i(1/{s_t}\,\mbf v_t \mbf v_t^\intercal)) + \numberthis\label{eq:encrypted_comp_vector2}\\
	&\quad+\rho \left( \mr{E_{v0}(\mbf v_t)} \odot (-\mbf e_i\odot\mr{E_{vr0}}(1/s_t)), S'+1\right),~i \in[S'+2]  \\
 	&\mr{E_{vr\ast}} ({\mbf Z_t}_{[j]}) = \mr{EvalSum}_{S+\bar T} \big ( \mr{E_{vr0}}(\mr{col}_{j}H\mbf Y_t )\odot \mr{e_{vr0}}(\lambda\mbf Q)\odot\mr{E_{vr0}}(\mbf{\bar y}_t, \mbf r_t)+ \mr{E_{vr0}}(\mr{col}_{j} H\mbf U_{t})\odot \mr{e_{vr0}}(\lambda\mbf R)\odot \mr{E_{vr0}}(\bar{\mbf u}_{t})  \big)\\ 
    &\mr{E_{v0}}(\boldsymbol\upsilon_{t,i}) = \sum_{k=0}^{S'+1} \sum_{j=0}^{S'+1} \mr{E_{v0}}(\mr{col}_k(\mbf M_{t}^{-1}))  \odot\left (\mr{E_{v0}}(\mr{row}_i \mbf U^f_{t}) \odot \mr{E_{vr\ast}}({\mbf Z_t}_{[j]}) \right),~i\in[m]\\
	&\mr{E_{v\ast}}(\boldsymbol\upsilon_{t}) = \sum_{i=0}^{m-1} \rho \left(\mr{E_{v0}}(\boldsymbol\upsilon_{t,i}, i(S+\bar T) \right). \\[0.5ex]
    &\text{After decryption:}\\[0.5ex]
	&{\mbf u_{t}}_{[i]} = \sum_{j=i(S+\bar T)}^{i(S+\bar T)+S'+1}  \boldsymbol\upsilon_{t,i},~i\in[m]. \numberthis\label{eq:encrypted_comp_vector3}
 \end{flalign*}

Following~\eqref{eq:encrypted_comp_vector}--\eqref{eq:encrypted_comp_vector3}, we get the depth expressions for $t\geq 1$:
 \begin{align}\label{eq:vector_packing_depth}
 \begin{split}
	d(\mbf M_t^{-1}) &= 2t + 3, \quad d(\mbf u_t) = d(\boldsymbol\upsilon_t) = 2t + 4.
 \end{split}
 \end{align}

\begin{remark}\label{rem:ringv0_vector}
	We make an abuse of notation when writing that we obtain $\mr{E_{v0}}(\cdot)$ and $\mr{E_{vr0}}(\cdot)$ after the rotation of a ciphertext with trailing zeros. Actually, the initial non-zero values will be shifted to the end of the encoded vector and we need to maintain a counter on how many times we can perform these rotations. 
	In the vector packing case, we need the parameters to satisfy the following rules, before a masking is necessary:
\begin{align}\label{eq:rules_vector}
\begin{split}
   \mr{ringDim/2} &>  \max(m,p)(N+M+\bar T)(S+\bar T)\\
   \mr{ringDim/2} &>  (S+\bar T)^2 \\
   \mr{ringDim/2} &> \max(mM,p(N+M))t(S+\bar T),
\end{split}
\end{align}
where the first line is for the online collection of samples ($\bar T$ is the total number of samples collected online), the second line is for the correct packing of the matrix into one ciphertext and the last line is the rule for after finalizing the collection of new samples $t\geq \bar T$ (instead of masking, here we can ask the client to prune the junk elements in $\bar{\mbf u}, \bar{\mbf y}$). 
\end{remark}

If the rules in Remark~\ref{rem:ringv0_vector} are not satisfied for the desired parameters of an application, the values can be split into the necessary number of ciphertexts (each packing up to $\mr{ringDim}/2$ values) and the computations can be readily extended to deal with this. 

\subsection{Proof of Lemma~\ref{lemma:s}}
\label{app:proof_precision}

The following identities can be found in~\cite[Ch.~1,3]{ben2003generalized}:
\begin{align}
&\lim\limits_{\delta\rightarrow 0^+} \mbf A^\ast(\delta \mbf I +\mbf A\mbf A^\ast)^{-1} =\mbf A^{\dagger}\label{eq:limit}\\
&\mbf A^\dagger = \mbf A^\ast {\mbf A^\dagger}^\ast \mbf A^\dagger.\label{eq:pseudo_id}
\end{align}

       (a) Using~\eqref{eq:limit}, we obtain:
\begin{align*}
    &\lim\limits_{\lambda_g\rightarrow 0^+} s = \mbf h^\intercal \mbf P \mbf h  -\mbf h^\intercal \mbf P^{1/2}\big(\lim\limits_{\lambda_g\rightarrow 0^+} \mbf P^{1/2} \mbf H (\mbf H^\intercal \mbf P\mbf H + \lambda_g\mbf I)^{-1}\mbf H^\intercal \mbf P^{1/2}\big) \mbf P^{1/2} \mbf h\\
    &= \mbf h^\intercal \mbf P^{1/2} \left( \mbf I - (\mbf P^{1/2} \mbf H) (\mbf P^{1/2}\mbf H)^\dagger \right) \mbf P^{1/2} \mbf h,
\end{align*}   
Furthermore, $\left( \mbf I - (\mbf P^{1/2} \mbf H) (\mbf P^{1/2}\mbf H)^\dagger \right)$ is orthogonal onto the range of $\mbf P^{1/2} \mbf H$. From the behavioral system definition (under the condition of persistancy of excitation) $\mbf P^{1/2} \mbf h \in \mr{Range}(\mbf P^{1/2} \mbf H)$, so $\lim\limits_{\lambda_g\rightarrow 0^+} s = 0$. 

       (b) We use the L'H\^ospital rule (because $\lim\limits_{\lambda_g\rightarrow 0^+} s = 0$):
\begin{align*}
&\lim\limits_{\lambda_g\rightarrow 0^+} \frac{s}{\lambda_g} = \lim\limits_{\lambda_g\rightarrow 0^+} \frac{\partial s/\partial \lambda_g}{\partial \lambda_g/\partial \lambda_g}  = 1 + \lim\limits_{\lambda_g\rightarrow 0^+} 
\mbf h^\intercal \mbf P \mbf H (\mbf H^\intercal \mbf P\mbf H + \lambda_g\mbf I)^{-2}\mbf H^\intercal \mbf P \mbf h \\
& = 1 + \mbf h^\intercal \mbf P^{1/2} (\mbf H^\intercal \mbf P^{1/2})^\dagger (\mbf P^{1/2}\mbf H)^\dagger \mbf P^{1/2}\mbf h,
\end{align*}  
where we used the pseudoinverse limit definition~\eqref{eq:limit}. 

       (c) The second part goes to 0 as $\lambda_g\rightarrow\infty$:
\begin{align*}
&\lim\limits_{\lambda_g\rightarrow \infty} \frac{s}{\lambda_g} = 1 + \lim\limits_{\lambda_g\rightarrow \infty} \frac{ \mbf h^\intercal \mbf P \mbf h}{\lambda_g} -\frac{\mbf h^\intercal \mbf P \mbf H (\mbf H^\intercal \mbf P\mbf H + \lambda_g\mbf I)^{-1}\mbf H^\intercal \mbf P \mbf h}{\lambda_g}.
\end{align*}

       (d) The proof is laborious so, for conciseness, we do not present the complete derivations. 
       We show that $f(\lambda_g)$ is decreasing by showing that $\frac{\partial f(\lambda_g)}{\partial \lambda_g}:= -\frac{g(\lambda_g)}{\lambda_g^2}$ takes only negative values, which we show by proving $\lim\limits_{\lambda_g\rightarrow 0^+} g(\lambda_g) \geq 0$ and $\frac{\partial g(\lambda_g)}{\partial \lambda_g} \geq 0$. To this end, we used again the pseudoinverse limit definitions and that $(\mbf P^{1/2}\mbf H)(\mbf P^{1/2}\mbf H)^\dagger$ is an orthogonal projection onto the range of $(\mbf P^{1/2}\mbf H)$ and hence, its eigenvalues are in $\{0,1\}$.
\hfill$\qedsymbol$

Let us investigate (b) further. Let $\mbf h = \mbf H\boldsymbol\eta$ for some vector $\boldsymbol\eta$, that in slowly-varying systems has a small magnitude. Then, by equation~\eqref{eq:pseudo_id}, we get that:
\begin{align*}
&\lim\limits_{\lambda_g\rightarrow 0^+} \frac{\partial s/\partial \lambda_g}{\partial \lambda_g/\partial \lambda_g} = 1 + \boldsymbol\eta^\intercal (\mbf P^{1/2}\mbf H)^\dagger (\mbf P^{1/2}\mbf H) \boldsymbol \eta.
\end{align*} 
$(\mbf P^{1/2}\mbf H)^\dagger (\mbf P^{1/2}\mbf H)$ is an orthogonal projection operator, i.e., its eigenvalues are in $\{0,1\}$, meaning that the quantity $\boldsymbol\eta^\intercal (\mbf P^{1/2}\mbf H)^\dagger (\mbf P^{1/2}\mbf H) \boldsymbol \eta$ has a small magnitude. Then, the second term in (b) is small.
\end{document}